\documentclass[prodmode,permissions]{acmsmall-ec14}
\usepackage[numbers]{natbib}

\usepackage{amsfonts, amsmath}
\usepackage[utf8]{inputenc}
\usepackage{graphicx}
\usepackage{bbm, comment}
\usepackage{subfigure}
\usepackage{color}
\usepackage{hyperref}
\usepackage[tworuled, noend, noline]{algorithm2e}

\newcommand{\cut}[1]{{}}

\providecommand{\abs}[1]{\lvert#1\rvert}
\newcommand{\nb}{\mathit{nb}}
\newcommand{\classone}{\mathbf{C}_\mathit{PNC}}
\newcommand{\classtwo}{\mathbf{C}_\mathit{TC}}
\newcommand{\loosecloudbound}{(1-p)^3 (1-\frac{p}{1-p})^2}

\title{Characterizing Strategic Cascades on Networks}
\author{TRAVIS MARTIN
\affil{University of Michigan}
GRANT SCHOENEBECK
\affil{University of Michigan}
MICHAEL P. WELLMAN
\affil{University of Michigan}
}

\begin{document}

\begin{abstract}
Transmission of disease, spread of information and rumors, adoption of new products, and many other network phenomena can be fruitfully modeled as cascading processes, where actions chosen by nodes influence the subsequent behavior of neighbors in the network graph.
       Current literature on cascades tends to assume nodes choose myopically based on the state of choices already taken by other nodes.
       We examine the possibility of strategic choice, where agents representing nodes anticipate the choices of others who have not yet decided, and take into account their own influence on such choices.
       Our study employs the framework of \citet{Chierichetti12}, who (under assumption of myopic node behavior) investigate the scheduling of node decisions to promote cascades of product adoptions preferred by the scheduler.
       We show that when nodes behave strategically, outcomes can be extremely different.
We exhibit cases where in the strategic setting 100\% of agents adopt, but in the myopic setting only an arbitrarily small $\epsilon$\% do.
Conversely, we present cases where in the strategic setting 0\% of agents adopt, but in the myopic setting $(100-\epsilon)\%$ do, for any constant $\epsilon > 0$.
       Additionally, we prove some properties of cascade processes with strategic agents, both in general and for particular classes of graphs.
\end{abstract}




\begin{bottomstuff}
A full version of the paper is available at: \url{http://arxiv.org/pdf/1310.2561.pdf}.

Author’s email addresses: \{travisbm, schoeneb, wellman\}@umich.edu.
\end{bottomstuff}

\maketitle
\thispagestyle{plain}
\pagestyle{plain}

\section{Introduction}
       \label{intro}

       A common topic in the study of network behavior is that of contagious or
\emph{cascading} processes, in which a number of nodes, or \emph{agents}, start with some
              property they then spread to their neighbors according to some specified propagation
              rules.
       This naturally represents phenomena such as the spread
              of trends, technologies, or influence among people or groups, or
              cascading failures in structures such as power grids or banks.
Scientists have, for many years, observed that processes can be heavily influenced by the network on which they occur \cite{Schelling71, Granovetter73, Granovetter78, Coleman88}.
              This influence has been confirmed in the real world by experiments from a wide array of fields \cite{Coleman57, Conley10, Lerman10, Banerjee12} including the study of product adoption \cite{Bass69, Brown87, Mahajan90, Goldenberg01}.

              Various models with simple spreading rules have been proposed
              \cite{Arthur89, Morris00, Watts02} to explain, for example, how
              breaking news spreads over the Internet or how a new technology spreads in popularity.
       Such models can be roughly classified in two categories, according to whether the spread is defined directly as a stochastic process, or in terms of decisions by self-interested agents who derive utility based on their choices and the choices of others in the network.
       In the latter case, the cascade scenario can be framed as a \emph{game}, and agent strategies cast as equilibria in the game.
       Due to the complexity of such games, however, typical agent-based cascade models  
       assume that agents make decisions \emph{myopically}, evaluating utility of alternative choices in the current state, without explicitly considering the future choices of others, nor their own potential impact on those choices.

       Our goal in this research is to investigate the implications of more forward-looking, or \emph{strategic} agent behavior.
       In what way do cascade patterns differ if agents behave strategically rather than myopically?
       How does the sophistication of agent decision making affect one's ability to influence a cascade process through scheduling of agent decisions?
       \subsection{Approach}
              To address these questions, we employ the framework of \citet{Chierichetti12}, described in Section~\ref{sec:model}.
              This prior work presents many interesting results about cascade behavior of myopic agents, and demonstrates the striking power of a scheduler to influence myopic cascades.
              Under our new assumption of strategic behavior, we find that even many simple cases of this game, such as pairwise agent interactions on a line, seem intractable to analyze. 
              Thus instead of solving the game generally, we take the approach of bounding the difference in cascade outcomes between myopic and strategic agent types.
              We find that cascade outcomes can be markedly different, as can the potential influence of a scheduler, depending on the particular network setting.
              We are able to obtain tight bounds through two easily analyzed graph families.

       \subsection{Results of Chierichetti et al.\ for myopic agents}
              \label{sec:myopic-description}

              Chierichetti et al.\ investigate a network of agents making choices between two options with positive externalities, $Y$ and $N$\@, under the influence of a scheduler.
              (We adopt their model and describe it in detail in Section~\ref{sec:model}.)
              The primary contribution of these authors is in analyzing the impact of the \emph{schedule}: the order in which agents make choices.
              They show that for any network there is some schedule which gets an expected constant fraction of the agents to choose $Y$\@.
              They also give networks and schedules which cause all but a constant number of agents to choose $Y$, in expectation.
              Lastly, they show that nonadaptive (fixed-sequence) schedules can obtain 50\% expected $Y$-adoption at best.

\subsection{Related work}
              \label{sec:related-work}

              Sequential voting and information cascades are two facets of a vast literature attempting to explain herd behavior \cite{Choi97}. Sequential voting models \cite{Alon12, Dekel00} consider strategic agents aiming to choose the majority decision, but with an additional private preference. Information cascades \cite{Banerjee92, Bikhchandani92} consider strategic agents with a noisy signal, attempting to determine the correct choice. Both of these models tend to simplify network effects by placing agents on a complete graph.

              \citet{Granovetter78} introduces the \emph{threshold model}, a foundational theory of network cascades which has since been studied and extended by many others \cite{Domingos01, Richardson02, Kempe03}.
              In the threshold model, agents take an action if a certain number of their neighbors have taken the same action.
              \citet{Altman12} give an example of self-interested agents behaving in accordance with the threshold model, but in general self-interested behavior may not align with set thresholds.

              Some agent-based cascade models \cite{Raub90, Morris00} allow agents to revise their decisions over multiple rounds of play.
              In each round, an agent myopically adopts its best choice in the current state.
              Some research on such models \cite{Blume93, Ellison93} also introduces an element of noise in agent choice, and investigates the convergence of cascades over time.

              \citet{Galeotti10} introduce a model with strategic agents which have access to incomplete information about the network outside their direct neighbors, thus making strategic agent behavior tractable.
              Lastly, \citet{Chierichetti12} introduce a cascade scheduling problem on networks based on a model studied by \citet{Arthur89}, which also assumes simple myopic agent decision making. Their model has been further extended by \citet{Cao13} and \citet{Hajiaghayi13}.


\section{Model}
       \label{sec:model}

       We model a \emph{game} $Q = (G, p, \pi)$ in which a collection of agents choose between two actions, $Y$ (``yes'') and $N$ (``no'').
       Agents make their choices one at a time in a sequence determined by the \emph{scheduler}.
       Once an agent has decided, it cannot change its action.
       We refer to the \emph{collection of agents} as $V$, the total number of agents as $n = \abs{V}$, an \emph{individual agent} as $i \in V$, and the \emph{choice} agent~\( i \) makes as $c_i \in \{Y, N\}$.
       Agents are vertices on the finite simple graph $G = (V, E)$, and we say that two agents are \emph{neighbors} if they are connected by an edge $e \in E$.
       We denote the set of neighbors of $i$ by $\nb(i) \subset V$.

       Each agent $i$ has a \emph{preference type}, $t_i \in \{Y, N\}$, which is independently randomly assigned at the beginning of the game.
       An agent is assigned type $Y$ with probability $p$ (a game parameter) and type $N$ with probability $1-p$. We assume that $Y$ is the less likely preference, so $p < .5$.
       Types are private: only $i$ knows the value of $t_i$ (until it is possibly revealed by $i$'s choice).

       Agents make their choices to maximize individual \emph{utility}.
       An agent obtains utility $\pi$ (a game parameter) for choosing its type ($c_i = t_i$), and a unit of utility for each neighboring node making the same decision that it does.
       Thus a node faces tension between choosing its type and the type it expects the majority of its neighbors to choose (when these types disagree).
       Formally, agent~$i$'s total utility is:
       \[u_i (t_i, c_i, \mathbf{c}_{-i} )
                =  \pi \mathbbm{1}(c_i = t_i) + \abs{ \{j \in \nb(i) : c_j = c_i\} },\]
       where $\mathbbm{1}$ is the indicator function and vector $\mathbf{c}_{-i}$ represents the choice of all other nodes.

       We differentiate between two modes of agent decision making: \emph{myopic} and \emph{strategic}.
       At the time agent~$i$ makes its decision, some nodes have already chosen and the remainder are undecided.
       A myopic agent makes its decision based on only the choices of decided nodes.
       It does not look into the future to consider the likely actions of undecided nodes, hence the term ``myopic''. Let $m_Y(i)$ and $m_N(i)$ denote the number of neighbors of~$i$ who have chosen $Y$ and $N$, respectively, at the time $i$ is scheduled to decide.
       Then a myopic $i$ chooses $t_i$ if $\abs{m_Y(i) - m_N(i)} \leq \pi$, and the majority type among its decided neighbors otherwise.

       A strategic agent aims to maximize its expected utility at the end of the game.
We assume it knows the details of the game ($G$, $p$, and $\pi$), the schedule $S$ (discussed below), and the decisions of already-decided agents.
       The agent reasons about the likely choices of undecided agents, assuming they all are strategic and play according to a \emph{perfect Bayesian equilibrium} (PBE)\@.
       A profile of strategies accords with PBE if and only if there exists a belief system, consistent with Bayesian updating, such that each agent's strategy is a best response to the other-agent strategies at every reachable information set.
       In our setting, each agent moves exactly once and types are independent, so there is no relevant updating.
       Under these conditions, each node of the game tree is essentially a singleton information set, treatable as a \emph{subgame}.
       Thus, the PBE concept here corresponds exactly to game solution by backward induction.
       To determine an agent's utility-maximizing action in some game, one can first solve for the choice of the last agent to move, in all possible subgames with only one agent left to move.
       Knowing the choice of the last agent, one can solve for the choice of the penultimate agent, in all subgames where all agents but two have moved.
       This reasoning can be repeated until the behavior of all agents in all subgames is known, yielding a PBE\@.

       We make the additional assumption that an agent chooses its preference type, $t_i$, if it would otherwise be indifferent between options.
       We show that any $Q$ and schedule combination correspond to exactly one PBE (see Theorem~\ref{prop:unique-spe}) consistent with this assumption.
       Thus the behavior of all strategic nodes is well defined.

       Following \citet{Chierichetti12}, our analysis includes a scheduler whose goal is to determine a schedule $S$ that maximizes the expected number of agents choosing $Y$\@.
       A schedule determines the order in which agents make their decisions.
       We consider two classes of schedule: \emph{nonadaptive} and \emph{adaptive}.
       A nonadaptive schedule is simply a fixed ordering of nodes, that is, a permutation of $V$\@.
       An adaptive schedule, in contrast, can select the next agent to choose based on previous agent decisions.
       Formally, adaptive schedule $S$ is a function of agent choices, $S : \{Y, N, U\}^{n}\allowbreak\rightarrow\allowbreak V$, where $U$ indicates that the corresponding agent is as yet \emph{undecided}.

       We evaluate schedules by their \emph{performance}, which is the expected number of nodes choosing $Y$ once all have decided.
       An \emph{optimal} schedule has the greatest performance among all schedules, or \emph{optimal performance}.
       We use \emph{strategic} and \emph{myopic} to qualify performance. For example, a schedule's strategic performance is the performance of the schedule for strategic agents.
       A state of a game in progress, in which some but not necessarily all agents have decided, is a \emph{situation}.

       We say that a situation is a $Y$-\emph{cascade} if every future agent chooses $Y$ regardless of type.
       We similarly define an $N$-\emph{cascade}.
       A situation is a \emph{total cascade} if the first agent necessarily initiates a cascade of its type.
       A game is a \emph{predetermined $Y$-cascade} if the starting situation is a $Y$-cascade.
       We similarly define \emph{predetermined $N$-cascade}.


\section{Roadmap of Results}
       \label{sec:contributions}
       The main result of this paper is a demonstration that cascade outcomes can vary drastically depending on the assumption of myopic or strategic agents.
Specifically, we show that the difference in performance between myopic and strategic agents can be arbitrarily close to the maximum possible difference of 100\% in either direction.
       In addition, we solve for equilibrium agent behavior in several particular game classes, provide miscellaneous results characterizing the behavior of cascade games with strategic agents, and a result demonstrating the importance of the capabilities of the scheduler:
       \begin{itemize}
              \item In Section~\ref{sec:clique}, we analyze the clique---both as a first example and as a way of introducing intuition, techniques, and results useful for subsequent sections.
We present instances in which strategic performance is 0\%: strictly worse than the constant expected adoption guaranteed for myopic agents.
We conversely present instances for which strategic performance is greater than myopic performance.
              \item In Section~\ref{sec:modified-clique}, we show that myopic performance can be much larger than strategic performance: the difference can be arbitrarily close to 100\%. We prove this by analyzing a specific class of games which occur on a graph we call a council graph.
              \item In Section~\ref{sec:cloud}, we show the converse: strategic performance can be arbitrarily close to 100\% greater than myopic performance. We show this by analyzing a class of games which occur on a graph we call a cloud graph.
              \item In Section~\ref{sec:general-results}, we give several results.  In particular we show that performance in the nonadaptive setting is bounded by $p$ for both myopic and strategic agents. This improves upon the results of \citet{Chierichetti12} showing a myopic agent upper bound of $\frac{1}{2}$. We also demonstrate a family of graphs in which myopic performance is always at least as great as strategic performance, no matter the parameter settings.
              \item In Section~\ref{sec:stack}, we investigate the commitment power of the scheduler and show that, in some cases, an ability to make non-credible threats can strictly enhance performance.
              \item  Finally, in Section~\ref{sec:block}, we present an algorithm to compute the performance of a graph with strategic agents that is efficient on a certain class of highly symmetric graphs.
       \end{itemize}

       Due to space limitations, the complete analysis of some results and several lemma and theorem proofs are relegated to the appendix.


\section{Clique Analysis}
       \label{sec:clique}
       A \emph{clique} is a complete graph where every two agents are connected.
       We begin our study of the difference between strategic and myopic performance with a description of behavior on the clique because it is illustrative of the difference between the myopic and strategic settings, and is used in subsequent proofs.
The clique is also easier to analyze as nodes occupy indistinguishable positions in the network, rendering all schedules identical.

\subsection{An example}
       \label{sec:clique-strategic-worse}
       Let $\pi = 1.1$, $p = 0.09$, and our graph be a clique of size~$3$.
       We name the nodes in the order that they are scheduled: $1$, $2$, and $3$.
       Note that on a clique all nodes have the same neighbors, so all schedules are identical.
       We reason about the behavior of strategic agents in this game by backward induction.

       First consider the behavior of the last node to choose, agent~3.
       If agents 1 and 2 have both chosen $N$ or have both chosen $Y$, 3 will match with them. Otherwise, $c_3 = t_3$.

       Next consider the behavior of agent~2.
       If $t_2 = N$, $c_2 = N$ no matter what. Even if $c_1 = Y$, agent~2 can expect to get a match from~3 with probability $0.91$ if $c_2 = N$.
       Its expected payoff would be $2.01$ for $c_2 = N$ versus $2$ for $c_2 = Y$\@.
       If $t_2 = Y$, $c_2 = N$ if $c_1 = N$. This is because agent~2's expected payoff for $c_2 = N$ is 2, versus $1.19$ for $c_2 = Y$.

       Knowing this behavior, $c_1 = N$ regardless of $t_1$.
       Suppose $t_1 = Y$.
       Then agent~1 gets payoff $2$ for $c_1 = N$, or payoff $1.1 + 0.91 (0 + 0.09) + (0.09) 2 \allowbreak= 1.3619$ for $c_1 = Y$, so is best off choosing $N$\@.

       Even this very simple graph demonstrates a qualitative difference between strategic and myopic behavior.
       As \citet{Chierichetti12} show, the optimal schedule for any graph with myopic agents achieves at least a constant fraction of $Y$-adoption, in expectation.
       In this example, myopic agents achieve over 6.7\% expected adoption.
       Yet for strategic agents, the example scenario yields zero adoption.
       We further characterize the behavior of the clique graph in Section~\ref{sec:clique-known}.

\subsection{Asymptotically large clique}
              \label{sec:clique-known}
       We characterize the behavior of games on cliques in the limit of large clique size.
       For the remainder of this section, we assume $p$ and $\pi$ to be fixed and represent a game $Q = (G,p,\pi)$ solely by its graph $G$. When $G$ is a complete graph (clique) of size $n$, we use $K_n$.
       Games on cliques can be divided into two classes of asymptotic behavior.

       \begin{theorem}
              \label{thm:clique-classes}
              For any fixed $0<p<1/2$ and $\pi >0$, there exists an $M$ such that for all $n \geq M$, $K_n$ gives either:
              \begin{enumerate}
                     \item A predetermined $N$-cascade (all agents always choose $N$), or
                     \item A total cascade (first agent chooses its type $t$ and the remaining agents match $t$, starting a $t$-cascade).
              \end{enumerate}
       \end{theorem}
       We denote these two classes of behavior by $\classone$ and $\classtwo$.
       The class a particular game belongs to depends on $p$, $\pi$, and $n$.
       The proof of this theorem, in Appendix~\ref{app:clique}, follows from the fact that, as cliques become very large, a node prefers any guaranteed cascade over a chance of being left out of a cascade.

       We find cliques in both $\classone$ (see Section~\ref{sec:clique-strategic-worse}) and $\classtwo$ (see below).
       $\classone$ corresponds to cases where myopic agents give higher performance than strategic agents, and $\classtwo$ corresponds to the opposite.
       A clique transitions from $\classtwo$ to $\classone$ as $p$ decreases and $\pi$ increases.
       On the boundary of this transition we find cases where a clique alternates, depending on the parity of $n$, between $\classone$ and $\classtwo$. Computational confirmation of these results can be found in Section~\ref{sec:clique-simulation}.

       \paragraph*{\textbf{When Strategic Outperforms Myopic on a Clique}}
              \label{sec:clique-strategic-better}
              Section~\ref{sec:clique-strategic-worse} presents an example where strategic agents yield zero performance but myopic agents give positive performance.
              One might expect the clique to always favor myopic performance, as strategic agents are aware that $Y$-preference is less likely, and thus might be more likely to choose $N$ than their myopic counterparts.
We show that this is not the case.
              When $1 \leq \pi < 1+p$, myopic agents underperform strategic agents because two $Y$ decisions are required to start a myopic $Y$-cascade and only one is required to start a strategic $Y$-cascade. Thus the probability of a $Y$-cascade is $\frac{p^2}{p^2 + (1-p)^2} \approx p^2$ for myopic agents and $p$ for strategic agents.

              \begin{theorem} \label{thm:clique-strategic-cascade}
                     For any $1 \leq \pi < 1 + p$, the probability of a $Y$-cascade with strategic users on a clique graph is $p$.
              \end{theorem}

              Note that when $\pi < 1$, strategic performance is equal to myopic performance by Lemma~\ref{lem:clique-pi-lessthan-one}.

              Computational results in Section~\ref{sec:clique-simulation} suggest that for any $\pi \geq 1$, there exists settings of $p$ such that strategic performance is greater than myopic performance.

       \subsection{Clique computational solution}
              \label{sec:clique-simulation}
              In Section~\ref{sec:clique-strategic-worse} we prove that, under some parameter settings for the clique, strategic agents result in a performance of zero.
              In Section~\ref{sec:clique-strategic-better} we prove that other parameter settings result in strategic agents outperforming myopic agents.
              In this section we provide computational verification for these two scenarios.
              For the specifics of our algorithm, please refer to Section~\ref{sec:block}.

              \begin{figure}
                    \centering
                    \subfigure[$n = 40$]{\includegraphics[width=0.493\columnwidth]{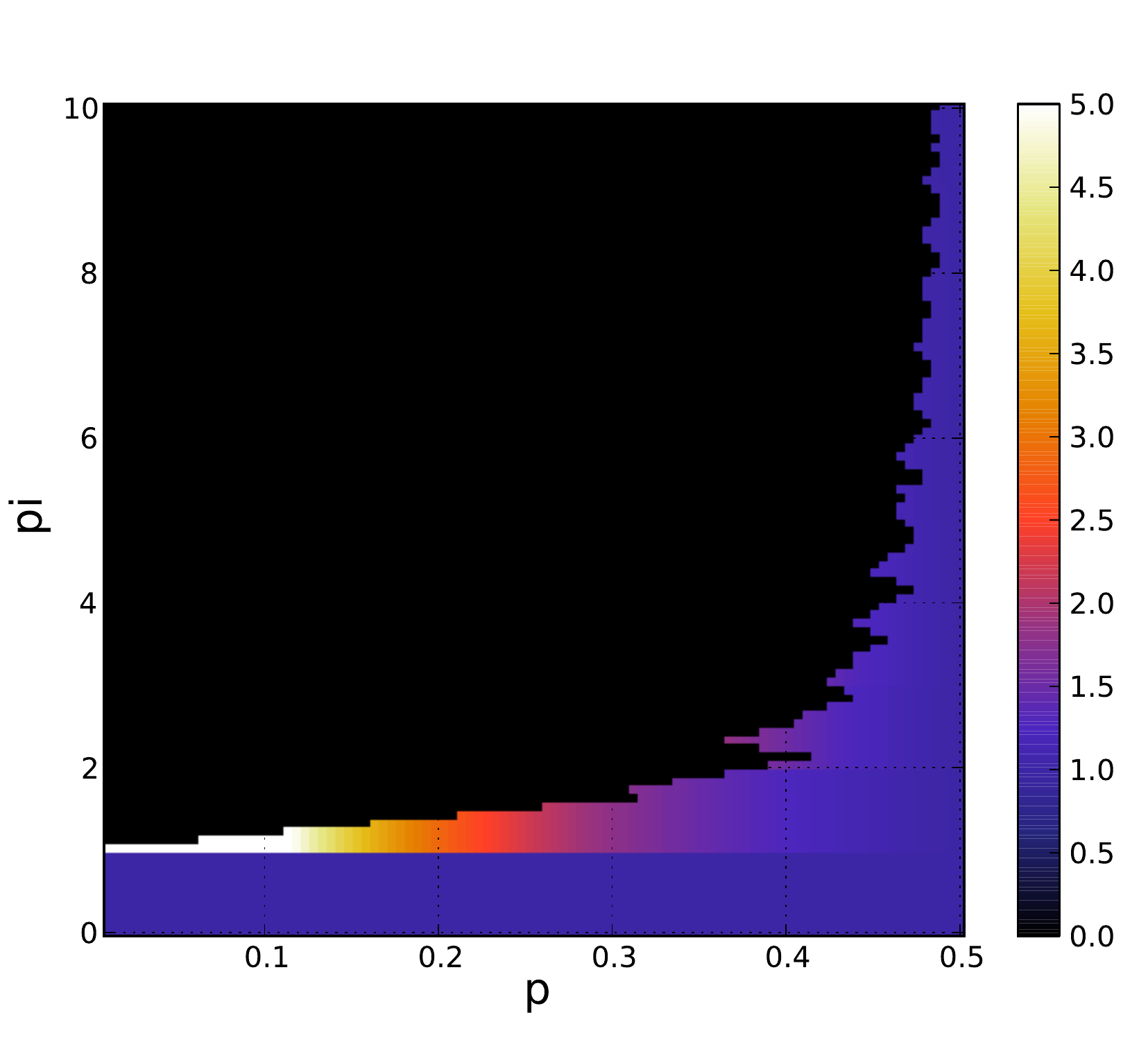}\label{fig:clique-simulation-1}}
                    \subfigure[$p = 0.25$]{\includegraphics[width=0.493\columnwidth]{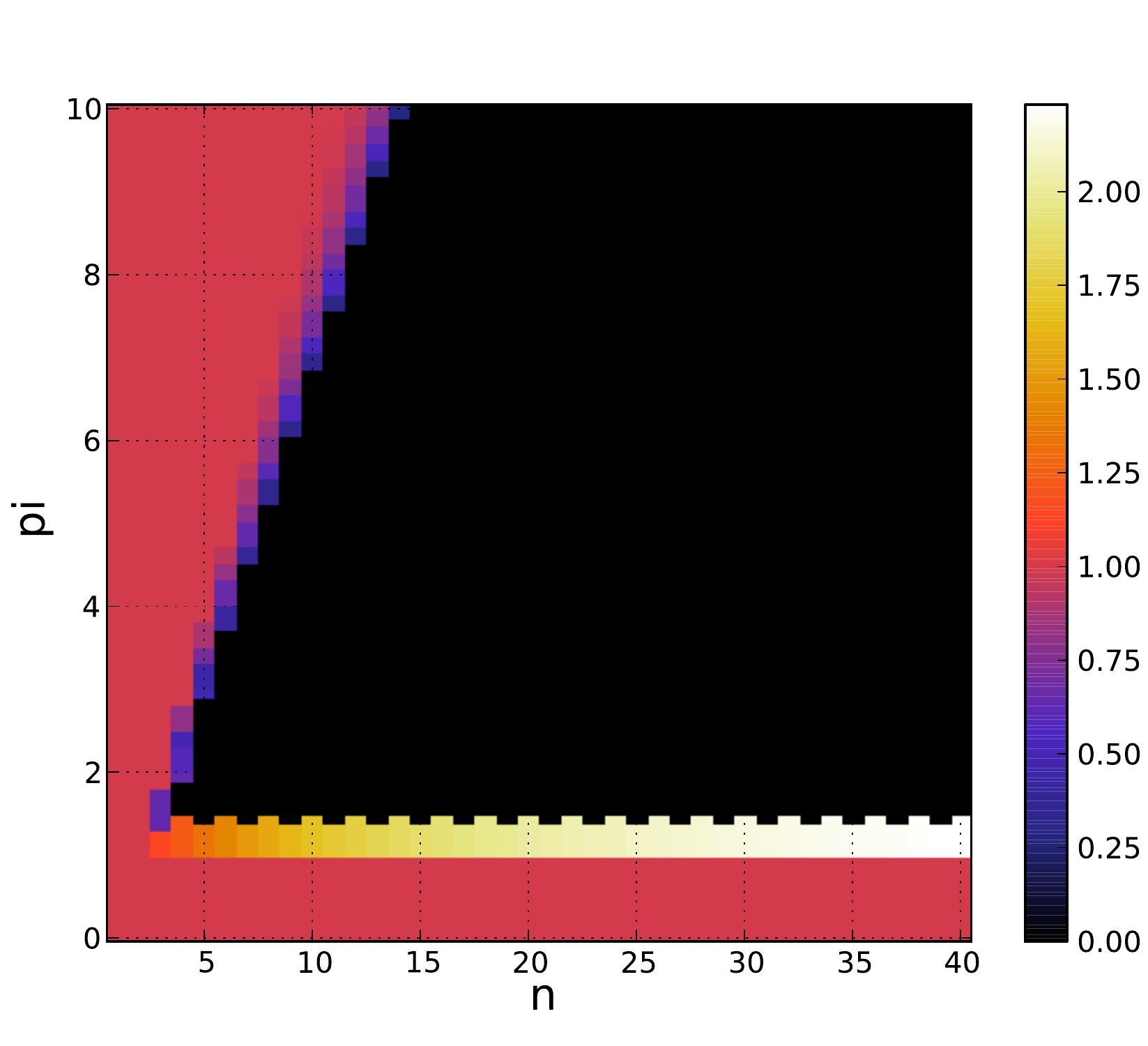}\label{fig:clique-simulation-2}}
                    \caption{Strategic-to-myopic performance ratio on clique graph.}
                    \label{fig:clique-simulation}
              \end{figure}
              Figure~\ref{fig:clique-simulation-1} displays results of a program which simulates a clique of 40 agents, each making the optimal strategic or myopic decision. We calculate the strategic and myopic performance for a variety of $p$ and $\pi$ combinations and plot their ratio in the figure. The black region corresponds to the class $\classone$ and the lighter regions correspond to the class $\classtwo$. The band at the bottom for $\pi < 1$ results from the immediate total cascade (Lemma~\ref{lem:clique-pi-lessthan-one}). The band just above $\pi=1$ in Fig.~\ref{fig:clique-simulation-1} corresponds to the region partially described by Theorem~\ref{thm:clique-strategic-cascade}, where one agent can start a strategic cascade but two agents are necessary for a myopic cascade.
              Figure~\ref{fig:clique-simulation-2} displays results of the same program, but with fixed $p$ to examine the effect of varying $n$.
              The resulting black area is governed by two simple bounds. The lower pink area results from $\pi < 1$ according to Lemma~\ref{lem:clique-pi-lessthan-one} as described above. The left pink wedge appears when $\pi$ is large relative to $n$ and all agents choose their preference. At the pink-black border we see non-trivial behavior.



\section{Myopic Outperforms Strategic}
       \label{sec:modified-clique}

We exhibit a setting where myopic performance is $(100- \epsilon)\%$ but strategic performance is zero.
Thus, unlike in the myopic case, where performance is always bounded above some constant~\cite{Chierichetti12}, it is possible to get zero strategic performance while simultaneously having arbitrarily high myopic performance. This constitutes the first half of our core result.
       We prove this bound constructively, by characterizing the behavior of a family of graphs which have optimal strategic performance of 0\% and an optimal myopic performance which approaches 100\% in the limit of large graph size.

       \begin{figure}
         \centering
         \includegraphics[width=0.25\columnwidth]{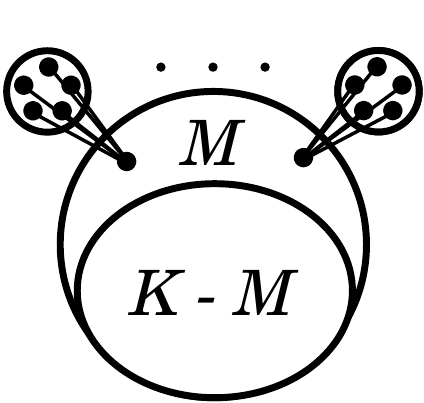}
         \caption{A council graph, as described in the text, with intra-clique connections excluded. The large circle is the council and the small circles are subcliques.}
         \label{fig:council_graph}
       \end{figure}

       Our graph is a modified version of a clique graph, which we call a \emph{council} graph (see Figure~\ref{fig:council_graph}).
       It consists of a large clique, the council, of size $K$ and $M$ smaller subcliques of size 5.%
\footnote{Any subclique of constant size $\geq 5$ will work.}
       \( M \) is $o(K)$, for example, $M = \sqrt{K}$.
       Each of the subcliques is completely connected to a unique node, its \emph{representative}, from the council.
       This gives $K-M$ council nodes of degree $K-1$, $M$ council nodes of degree $K+4$, and $5M$ subclique nodes of degree $5$.


       \paragraph*{\textbf{Near 100\% Myopic Performance}}
              \label{sec:council-myopic}
              We demonstrate a schedule giving performance tending to $100\%$ in the limit of large graph size.
              We do not prove this schedule's optimality, but it gives a lower bound sufficient for our purposes. We say a subclique is \emph{fresh} if none of its nodes have been scheduled. Our schedule, $S$, is the following:

              \RestyleAlgo{tworuled}
              \begin{algorithm}[H]
              \nl Choose any fresh subclique, $j$. \\
                            \Indp \Indp Schedule nodes from $j$ until one chooses $N$ or all have chosen $Y$. \\
                            If all nodes in $j$ have chosen $Y$, schedule $j$'s council representative, $r_j$. \\ \Indm \Indm
              \nl Repeat \texttt{1} until three representatives have been scheduled or no fresh subcliques remain. \\
              \nl Schedule all council nodes without $N$-decided neighbors. \\
              \nl Schedule all remaining council nodes in any order. \\
              \nl Schedule all remaining subclique nodes in any order.
              \end{algorithm}%
              \begin{theorem}
                     For any $p < .5$, $2 < \pi < 3$, the myopic performance of $S$ approaches 100\% as $K \rightarrow \infty$.
              \end{theorem}%
              \begin{proof}
                     Our proof proceeds by a careful description of behavior at each point in the schedule. First note that a myopic agent will choose its type if $\leq 2$ of its neighbors have been scheduled, and $Y$ if $\geq 3$ more of its neighbors have chosen $Y$ than $N$.

                     So with probability $p^2$ a fresh subclique from Step~\texttt{1} will be a $Y$-cascade and its representative will also choose $Y$\@. 
With probability $1-p^2$ a fresh subclique from Step~\texttt{1} will not be a $Y$-cascade and its representative will not be scheduled.

                     We can bound the probability of (the undesirable event of) not having 3 subclique $Y$-cascades by:
                     \begin{multline*}
                     \binom{M}{2}(p^2)^2 (1-p^2)^{M-2} + M p^2 (1-p^2)^{M-1} + (1-p^2)^M < M^2 (1-p^2)^{M-2}.
                     \end{multline*}

                     Once there are 3 subclique $Y$-cascades, the entire council chooses $Y$. Thus, the expected fraction of $Y$ is at least: $$K (1 - M^2(1-p^2)^{M-2}) / (K+M),$$ which tends to 1 as $K \rightarrow \infty$. \qed
              \end{proof}

       \paragraph*{\textbf{0\% Strategic Performance}}
              We characterize behavior of the council graph with strategic agents and prove that certain choices of $p$ and $\pi$ give 0\% performance. This, together with the results from above, gives a tight bound on the extent to which myopic performance can be greater than strategic performance.

              \begin{theorem}\label{thm:councilstrategicbad}
                     For some $p < .5$, $2 < \pi < 3$, any schedule on a council graph with strategic agents has 0\% performance.
              \end{theorem}

          The following lemma is used in the proof below.
              \begin{lemma}
                     \label{lem:council-behavior}
                     For some $p < .5$, $2 < \pi < 3$, a clique of $k \geq 5$ undecided nodes and one node guaranteed to choose $Y$ will result in $k$ $N$-decisions and 1 $Y$-decision.
              \end{lemma}
              We prove this lemma in Appendix~\ref{app:council} by direct comparison of expected utilities.

              \begin{proof}[of Theorem~\ref{thm:councilstrategicbad}]
                     The council graph was chosen to facilitate analysis by simplification to more easily understood cliques. As such, we invoke Lemma~\ref{lem:council-behavior}, which proves the existence of cliques and settings of $p$ and $\pi$ which strongly favor $N$, in the sense that even if a common neighbor of the clique is guaranteed to choose $Y$, the bias towards $N$-preference results in a predetermined $N$-cascade.

                     Thus, even if a clever scheduler convinces the whole council to choose $Y$, we see by Lemma~\ref{lem:council-behavior} that, for some $p$ and $\pi$, the subclique chooses $N$.

The nodes in the council, being fully strategic, know any subclique neighbors they have are guaranteed $N$-neighbors.
By another application of Lemma~\ref{lem:council-behavior}, we see that all council nodes choose $N$\@ and thus, for some $p$ and $\pi$, any schedule is doomed to $0\%$ performance. \qed
              \end{proof}


\section{Strategic Outperforms Myopic}
       \label{sec:cloud}

By now it is natural to see how strategic agents' expectations of future $N$-preference lead to lower strategic performance than myopic performance.
As seen in Section~\ref{sec:clique-strategic-better}, there are also games where strategic agents give higher performance than myopic agents.
In this section we show the second half of our core result, that this difference can be as large as $(100-\epsilon)\%$.

       We prove this bound constructively by analyzing a special graph we call a \emph{cloud} graph. We first show that for certain parameters on the cloud graph it is possible to obtain strategic performance of 100\%.  Recall that no graph can achieve exactly 0\% or 100\% myopic performance, because the first myopic node always chooses its type.
We do, however, show that for some settings on the cloud, myopic performance approaches $0\%$ while strategic performance remains $100\%$, giving our desired bound.

       A cloud graph (Figure~\ref{fig:cloud_graph}) consists of two singular \emph{outer} vertices of
       degree $a$ and $b$ respectively, one singular \emph{inner} vertex of
       degree $a + b$, and two \emph{clouds} of vertices of respective size $a$ and $b$, with
       each vertex of degree two. Each of the outer vertices is connected
       to every vertex in a distinct cloud. The inner vertex is
       connected to every vertex in both clouds.
We call the cloud with $a$ vertices $A$ and the one with $b$ vertices $B$.

       \begin{figure}
         \centering
         \includegraphics[width=0.30\columnwidth]{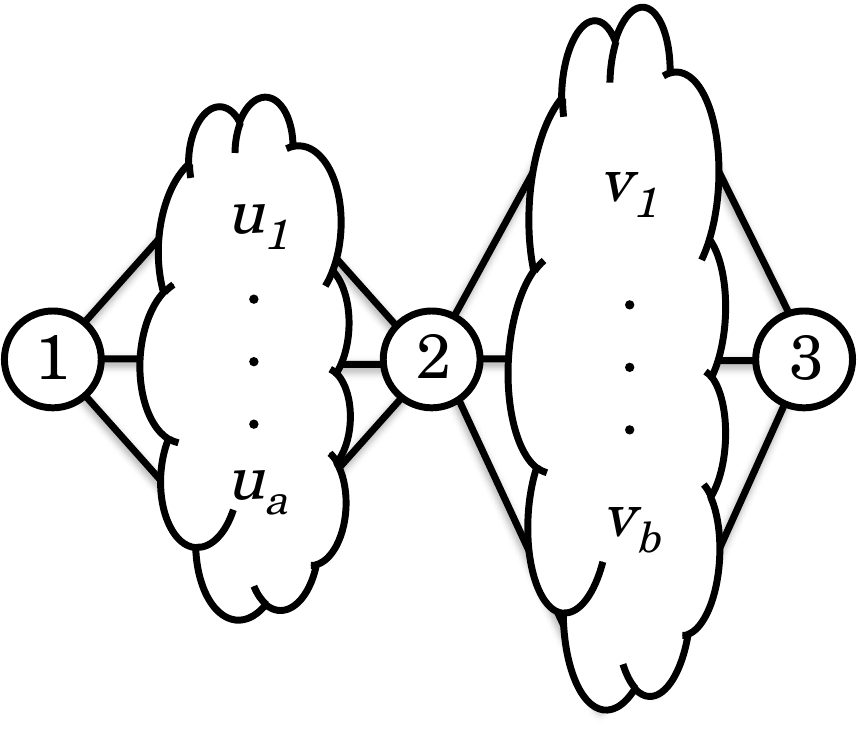}
         \caption{An example cloud graph.}
         \label{fig:cloud_graph}
       \end{figure}

        \begin{theorem}
        Fix arbitrary $\epsilon > 0$.  Then there exist parameters $a$, $b$, $p$, and $\pi$ such that strategic performance on the cloud graph is 100\% whilst the myopic performance is at most $\epsilon$.
        \label{thm:cloud-100}
        \end{theorem}
        This follows from Lemma~\ref{lem:cloud-onehundred} and Lemma~\ref{lem:cloud-small-myopic}, proved in the remainder of this section.

       \paragraph*{\textbf{100\% Strategic Performance}}
              We give sufficient conditions for obtaining 100\% performance with strategic agents in a cloud graph. We refer to the singular vertices, from left to right, as $1, 2$, and $3$, and assume that $a < b$.

              An optimal schedule, $S_{opt}$, is the following:

              \begin{algorithm}[H]
                     \nl Schedule 1. \\
                     \If{1 chooses $Y$}{Schedule 2 followed by 3.}
                     \Else{Schedule 3 followed by 2.}
                     \nl Schedule all cloud vertices in any order.
              \end{algorithm}

              An overview of the proof of optimality is as follows. Scheduling agent~2 before agent~3 guarantees that 2, and thus all nodes in $A$, will match the choice of 1 (Lemma~\ref{lem:2-chooses-1}). On the other hand, scheduling agent~3 before agent~2 gives some positive probability that the nodes in $A$ choose their type (Lemma~\ref{lem:3-chooses-type}). This outcome results in lower utility for 1. Thus the adaptive schedule can be used to incentivize 1 to choose $Y$ through threat of punishment for choosing $N$\@. We are able to show that, for large enough cloud sizes, threat of punishment to 1 for choosing $N$ is enough to convince it to choose $Y$, giving 100\% performance. This is given formally below.
              \begin{lemma}
                     \label{lem:cloud-onehundred}
                     If cloud sizes satisfy $a (1-p) + \pi < b p$ and $ap^2 > \pi$, then, under the schedule $S_{opt}$, 1, and thus all agents, will always choose $Y$\@.
              \end{lemma}
              \begin{proof}
              We are able to punish 1 for choosing $N$ because scheduling 3 first gives a $p^2$ chance of the nodes in $A$ choosing their type, while scheduling 2 first guarantees all nodes in $A$ choose $c_1$, a more desirable outcome to 1. We use several lemmas outlining the behavior of agents 2 and 3, proved below.

              Assume 1 is $N$-type. Being $Y$-type only increases 1's payoff for choosing $Y$\@. We show below that 1's utility for choosing $Y$ is higher than its utility for choosing $N$ when $a (1-p) + \pi < b p$ and $ap^2 > \pi$.

              If 1 chooses $N$, then we schedule 3, followed by 2. In this case, Lemma~\ref{lem:3-chooses-type} shows that $c_2 = c_3 = t_3$. With probability $1-p$ all nodes in $A$ choose $N$, and with probability $p$ $A$ is split. 1's expected utility is $\pi + a (1-p) + a (1-p) p$.

              If 1 chooses $Y$, then we schedule 2, followed by 3. In this case, Lemma~\ref{lem:2-chooses-1} shows that $c_2 = 1$. 1's expected utility is $a$.

              Then 1 will choose $Y$ as long as $a > \pi + a(1-p)(1+p)$. Or, equivalently, $a p^2 > \pi$, which is true by assumption. Once $1$ chooses $Y$, the schedule leads to all remaining nodes choosing $Y$. \qed
              \end{proof}

              We must pick appropriate cloud sizes (depending on $\pi$ and $p$) and have $\pi < 2$ for the theorem to be true. This is possible for any $p$ by selecting large enough $a$ and even larger $b$. The following lemmas detail the behaviors of the clouds and agents 2 and 3 used in the proof of Lemma~\ref{lem:cloud-onehundred}.

              \begin{lemma}
                     The behavior of an unscheduled cloud neighbored by two decided singular agents is completely determined by the singular agents' choices.
                     If their choices are different, then every cloud agent will choose its type and an expected $p$ fraction of the cloud agents will choose $Y$. In this case we say that the cloud has been \emph{split}.
                     If they make the same choice $c$, all cloud agents will choose $c$.
              \end{lemma}

              Knowing the cloud behavior, we can characterize the behavior of the case where agent~3 is scheduled and then agent~2 is scheduled.
              \begin{lemma}
                     \label{lem:3-chooses-type}
                     When $a (1-p) + \pi < b p$, if agent~3 is scheduled to choose and 2 has not been scheduled yet, $c_3 = t_3$. Then, when 2 is scheduled next, $c_2 = c_3$.
              \end{lemma}

              \begin{lemma}
                     \label{lem:2-chooses-1}
                     If agent~2 is scheduled to choose after 1 but before 3, it will choose $c_1$ if $bp > ap > \pi$. When 3 is scheduled, it will match 2.
              \end{lemma}

Here, the scheduler persuades cloud agents to adopt the minority preference through the threat of unfavorable adaptive sequencing.
All nonadaptive schedules have performance bounded by $p$ (Theorem~\ref{thm:non-adpative-p}), and thus the above result clearly requires adaptivity.
In fact, the difference in performance for adaptive and nonadaptive scheduling can be arbitrarily large for strategic agents (Corollary~\ref{cor:adaptive-greaterthan-nonadaptive}).

       \paragraph*{\textbf{Near 0\% Myopic Performance}}
              \label{sec:cloud-bound}
              By Lemma~\ref{lem:cloud-onehundred}, we can obtain 100\% adoption with strategic agents for any $p$ if we pick cloud sizes $a$ and $b$ large enough. The proportion of myopic adoption, however, is some polynomial of $p$, and thus can be made arbitrarily small. The combination of these two results gives us a tight bound on the extent to which strategic performance can be greater than myopic performance.

              \begin{lemma}
              \label{lem:cloud-small-myopic}
        Fixing $\pi < 2$, the  myopic performance is bounded by $p \loosecloudbound + [1-\loosecloudbound]$.  In the limit of $p \rightarrow 0$, the proportion of myopic adoption in the cloud graph also goes to 0.
        \end{lemma}

  \begin{proof}
              We bound the myopic performance by a polynomial in $p$ for $\pi < 2$.

              Denote the current difference between the number of agents in cloud $A$ ($B$) who have chosen $Y$ and those who have chosen $N$ by $d_A$ ($d_B$). With probability $(1-p)^3$, all singular agents are type $N$\@. A singular $N$-type agent will choose $Y$ only if $d_A \geq 1$ or $d_B \geq 1$. A cloud agent will choose its type or $N$ unless at least one of the singular agents has already chosen $Y$\@.

              We bound the probability of a singular agent choosing $Y$ by noticing that the probability of a cloud ever achieving a $Y$ majority by agents choosing their type is no greater than $\frac{p}{1-p}$, a result from the mathematics of biased random walks. Thus, the probability that all cloud nodes choose their type (or $N$) is at least $\loosecloudbound$. This probability, which we denote $q$, tends to 1 as $p \rightarrow 0$. The expected proportion of $Y$-adoptions is no greater than $p q + (1 - q)$, which goes to 0 as $p \rightarrow 0$.
        \qed
       \end{proof}


\section{Miscellaneous Results}
       \label{sec:general-results}
       Having analyzed behavior of cascades on specific classes of graphs, we aim to give properties of cascade behavior for arbitrary games, regardless of $G$, $p$, or $\pi$.

       We first address the issue of the multiplicity of PBE\@.
       The existence or uniqueness of PBE is not guaranteed for all classes of games.
       The possibility of zero or multiple PBE would render some of our key concepts, such as the performance of a schedule, unclear.
       Fortunately, as we show below, our assumption that agents consistently choose their type when indifferent between options always results in the selection of a unique PBE\@. This follows from a simple backward induction argument.
       We also give a technique for relaxing this behavioral assumption but keeping the unique PBE property:
       \begin{theorem}
              \label{prop:unique-spe}
              If agents are never indifferent between choices, or always resolve any indifference in a consistent way---by choosing the same option whenever they are in the same situation---then their PBE behavior is uniquely defined.
       \end{theorem}

       The following theorem shows that our assumption of indifference can be avoided, while keeping the same cascade outcome, by slightly adjusting $\pi$.
       \begin{theorem}
              \label{prop:remove-indifference}
              Given a game $Q = (G, p, \pi)$, let $P$ be the performance under an adaptive schedule $S$ with the assumption that a node always chooses its type if it is indifferent between $Y$ and $N$.
              Then there exists an $\epsilon$ such that $Q' = (G, p, \pi' = \pi + \epsilon)$ also achieves performance $P$ under $S$, and under $S$ no node is indifferent between choices.
       \end{theorem}
       The proof, in Appendix~\ref{app:misc}, simply chooses $\epsilon$ less than the smallest utility difference.

       We also prove that increasing $p$ alone can never decrease the performance of the optimal schedule. The main ingredient in the proof is a coupling argument.  This monotonicity is not observed for $\pi$ or $n$.%
\footnote{An example of non-monotonicity can be found in Figure~\ref{fig:clique-simulation-2}.}

       \begin{theorem}
              \label{prop:monotonic-p}
              For any two games, $Q = (G, p, \pi)$ and $Q' = (G, p', \pi)$, with $0\nobreak <\nobreak p\nobreak <\nobreak p'\nobreak <\nobreak .5$, the performance of any nonadaptive schedule $S$ for $Q$ is weakly worse than the performance of $S$ for $Q'$.
       \end{theorem}

\paragraph*{\textbf{Star Graph}}

       It seems that behavior on every graph we study varies unpredictably as game parameters change.
       Even a graph as simple as a clique can exhibit drastically different cascade outcomes from small changes in $p$ or $\pi$. However, this is not always the case.
Games on the \emph{star} graph---a graph with one \emph{interior} agent (of degree $n-1$) connected to $n-1$ \emph{exterior} agents (of degree 1)---have notably regular behavior.

       \begin{theorem}\label{theorem star} \label{thm:star-spe} For any parameters on any star graph, the optimal performance in the myopic setting is at least the optimal performance in the strategic setting for both adaptive and nonadaptive schedules.   \end{theorem}
       The proof, in Appendix~\ref{app:star}, establishes the optimality of threshold strategies for nodes and then shows that the myopic thresholds always beat the strategic thresholds.

	        	\label{sec:star-bounds}
	        	Knowing that myopic performance exceeds strategic on the star, we next explore the degree of this advantage.
	    We find that, for adaptive schedules, myopic performance can be arbitrarily close to an additive factor of $50\%$ greater than strategic performance.

	        	In the limit of large star graphs, adaptive myopic performance is $\frac{p}{1-p}$ for any $\pi < 1$. Thus, for $p$ arbitrarily close to .5, myopic performance approaches $100\%$. This result does not hold for strategic agents: a backward induction argument shows that for small enough $\pi$, strategic performance is bounded by $p$.
	        	\begin{theorem}
	        		\label{prop:star-50-strategic}
	        		For any star graph with $\pi < 1-p$ and any adaptive schedule, strategic $Y$-type nodes choose $N$ when a majority of nodes have chosen $N$, upper bounding strategic performance by $p$.
	        	\end{theorem}

\begin{proof}
	We characterize behavior on the star with strategic agents when $\pi < 1-p$ by backward induction. Let $d$ be the difference between the number of exterior nodes who have chosen $Y$ and the number who have chosen $N$: $d = m_Y - m_N$\@. We show that any node chooses $N$ when $d<0$.

	Consider the behavior of the last node, $i$. Assume the best case for a $Y$ choice, that $t_i = Y$. If $d=-1$, $i$ receives $p + \pi$ utility for $c_i =Y$ and $1$ utility for $c_i = N$. By assumption, $i$ prefers $N$\@.
Any $d<-1$ gives $i$ $\pi$ utility for $c_i = Y$ and $1$ utility for $c_i = N$\@.
This completes the base case.

	Next we prove the inductive step. Assuming the theorem holds when $k-1$ agents remain, we show the theorem holds when $k$ agents remain. Denote the agent choosing with $k$ agents remaining by $j$. Assume the best case for a $Y$ choice, that $t_i=Y$. Using the inductive hypothesis, we find that utilities for $j$ are exactly as above for $i$. \qed
\end{proof}

 We proved that optimal strategic performance is never greater than optimal myopic performance for a star graph. We also present computational verification. For details of the algorithm used for computing solutions, see Section~\ref{sec:block}.

              \begin{figure}
              \centering
              \includegraphics[width=0.493\columnwidth]{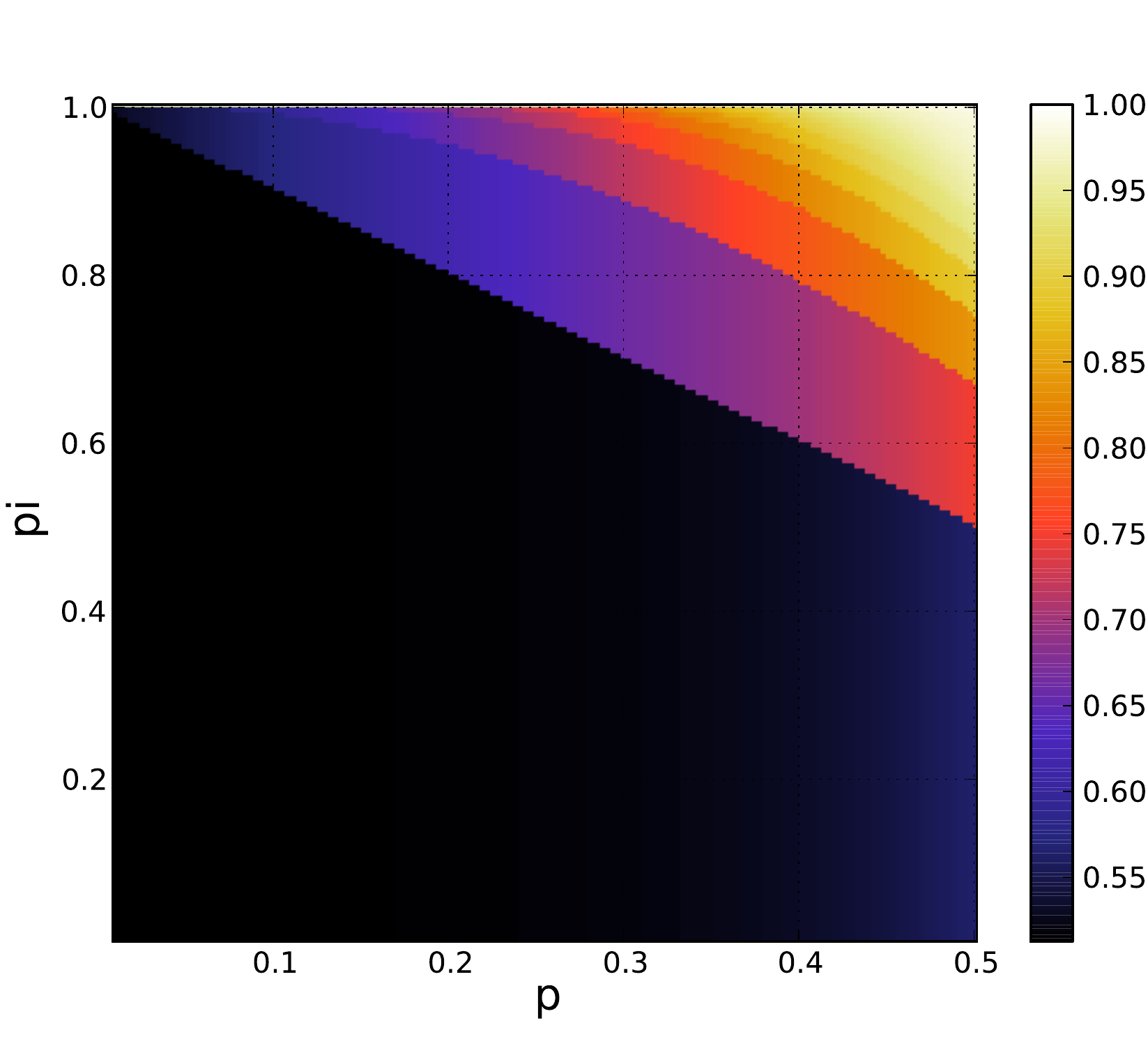}
              \caption{Strategic-to-myopic performance ratio on star, \( n=41 \).}
              \label{fig:star-simulation}
              \end{figure}

	            Figure~\ref{fig:star-simulation} displays results of a computational solution of the optimal schedule for a star of 41 agents. We calculate the strategic-to-myopic performance ratio for a variety of $p$ and $\pi$ combinations.

\paragraph*{\textbf{Nonadaptive Schedules}}

Lastly, we prove that no nonadaptive schedule for strategic or myopic agents can achieve more than a $p$ fraction performance, on any graph.
A similar bound of $p$ was proved independently by \citet[Theorem~1]{Hajiaghayi13}, using different techniques, restricted to myopic agents on the clique graph.
Their theorem generalizes to the setting where agents can have heterogeneous $\pi$ thresholds.  
Whereas we do not explicitly address heterogeneity in $\pi$ here, we note that the proof of Theorem~\ref{thm:non-adpative-p} immediately extends to this more general model.

Both results improve on the 50\% bound of \citet{Chierichetti12}, which covers myopic agents on arbitrary graphs.

Bounding nonadaptive schedule performance for strategic agents entails that the very high performance of Section~\ref{sec:cloud} is not possible when the scheduler cannot react to decisions made by nodes (Corollary~\ref{cor:adaptive-greaterthan-nonadaptive}).
Moreover, it rules out the possibility of predetermined $Y$-cascades with nonadaptive schedules.  Our proof combines a careful inductive argument with the repeated application of a result from the analysis of Boolean functions.

       \begin{theorem}
              \label{thm:non-adpative-p}
              No nonadaptive schedule can achieve more than $p$ fraction performance for any $p \leq .5, \pi > 0$, in the myopic or strategic setting.
       \end{theorem}

    To prove Theorem~\ref{thm:non-adpative-p} we use a lemma from \citet[Lemma~5.1]{MosselNT2012}:
    \begin{lemma}\label{lemma:mossel}  Let $f: \{Y, N\}^n \rightarrow \{Y, N\}$ be a monotone function (so that flipping input bits from $N$ to $Y$ cannot change the output from $Y$ to $N$ and vice versa)   with $P_{1/2}(f = Y) = 1/2$.  Then $P_p(f = Y) \leq p$ for all $0 \leq p < 1/2$.  \end{lemma}

    \begin{proof}[of Theorem \ref{thm:non-adpative-p}]
           We first prove the theorem for the myopic setting.  We apply Lemma~\ref{lemma:mossel} separately to each node, in combination with Lemma~\ref{lemma:a100}, to show that each agent chooses $Y$ with probability at most $p$.  From linearity of expectations, we know the myopic performance is at most a $p$ fraction of the nodes.

           Fix a game and schedule, and adopt the following notation.

           Let $c_i:\{Y, N\}^i \rightarrow \{Y, N\}$ be the function which takes as input the types of the first $i$ agents and outputs the selection of agent $i$.

           Let $\mathring{c}_i:\{Y, N\}^i \rightarrow \{Y, N\}^i$  be the function which takes as input the types of the first $i$ agents and outputs the selection of the first $i$ agents.

           Let $\hat{c}_i:\{Y, N\}^{i-1} \times \{Y, N\} \rightarrow \{Y, N\}$  be the function which takes as input the selections of the first $i-1$ agents and the type of the $i$th agent and outputs the selection of the $i$th agent.

           We denote the types of the first $i$ agents as $t^{(i)} =  t_1, \ldots, t_i \in \{Y, N\}^i$ and denote by $\neg w \in \{Y, N\}^i$ the string with each coordinate the opposite as in $w \in \{Y, N\}$.

           Lemma~\ref{lemma:a100} shows that $c_i(\neg  t^{(i)}) = \neg c_i(t^{(i)})$, from which we see that $P_{1/2}(c_i = Y) = 1/2$ because for each string, exactly one of $w$ and $\neg w$ evaluates to $Y$.  Thus we can employ Lemma~\ref{lemma:mossel} to see that $P_{p}(c_i = Y) \leq p$, which proves the theorem in the myopic case.

           \begin{lemma}\label{lemma:a100}
           $c_i(\neg  t^{(i)}) = \neg c_i(t^{(i)})$. \end{lemma}
           \begin{proof}[of Lemma~\ref{lemma:a100}]
               We can do this by induction on $i$ to show that both $c_i$ and $\mathring{c}_i$ have this property.   Note that because $Y$ and $N$ are treated symmetrically in the myopic setting, we know that $\hat{c}_i(\neg w) = \neg \hat{c}_i(w)$ for all $i$.

                  The base case follows because $c_1(t_1) = \mathring{c}_1(t_1) = \hat{c}_1(t_1)$, and we know that $\hat{c}_1$ has the property.  

                  Assume that the statement is true for all $j < i$. Note that
                  \begin{align*} c_i(\neg t^{(i)}) & = \hat{c}_i(\mathring{c}_{i-1}(\neg t^{(i-1)}), \neg t_i) \\
                                                    & = \hat{c}_i(\neg \mathring{c}_{i-1}(t^{(i-1)}), \neg t_i) \\
                                                    & = \neg \hat{c}_i(\mathring{c}_{i-1}(t^{(i-1)}), t_i) = \neg c_i( t^{(i)}).  \end{align*}

                  The first line follows from the definition of $c_i$ and $\hat{c}_i$ and second line follows from induction.
                  Similarly, for $\mathring{c}_i$:
                  \begin{align*} \mathring{c}_i(\neg t^{(i)}) & = \mathring{c}_{i - 1}(\neg t^{(i-1)}) \circ \hat{c}_i(\mathring{c}_{i-1}(\neg t^{(i-1)}), \neg t_i) \\
                                                    & = \neg \mathring{c}_{i - 1}(t^{(i-1)}) \circ \hat{c}_i(\neg \mathring{c}_{i-1}( t^{(i-1)}), \neg t_i) \\
                                                    & = \neg \mathring{c}_{i - 1}(t^{(i-1)}) \circ \neg \hat{c}_i(\mathring{c}_{i-1}( t^{(i-1)}),  t_i) = \neg \mathring{c}_i( t^{(i)}).\qed \end{align*}
           \end{proof}

           We next prove the strategic case of Theorem~\ref{thm:non-adpative-p}. The intuition is straightforward.
If a node imagines that all future nodes are equally likely to prefer $Y$ and $N$, then again $Y$ and $N$ are treated symmetrically, as in the myopic setting, and Lemma~\ref{lemma:mossel} applies.
So given that this node's type and the types of agents that have already chosen are $Y$ independently with probability $p$, the probability that each node chooses $Y$ is at most $p$.
This probability only decreases when this node expects future nodes to be $Y$-type less often.

           We define $c_i^p:\{Y, N\}^i \rightarrow \{Y, N\}$, $\mathring{c}_i^p:\{Y, N\}^i \rightarrow \{Y, N\}^i$, and $\hat{c}_i^p:\{Y, N\}^{i-1} \times \{Y, N\} \rightarrow \{Y, N\}$ analogously to above, except here we assume that all agents play strategically according to the case where each node is $Y$-type with probability $p$.

           The outline of the proof of Lemma~\ref{lemma:a200}, given in full in Appendix~\ref{app:misc}, is as follows.
            We again see that $P_{1/2}(c_i^{1/2}(t^{(i)}) = Y) = 1/2$ by the same reasoning, and applying Lemma~\ref{lemma:mossel} we see that $P_{p}(c_i^{1/2}(t^{(i)}) = Y) \leq p$.  We would like to show that $P_{p}(c_i^{p}(t^{(i)}) = Y) \leq p$.  To complete the lemma it is enough to show that $c_i^p$ is monotone in $p$.  That is, increasing $p$ only makes a $Y$ outcome more likely.

           By induction we will show that $c_i^p$ is also monotone with respect to $p$.  This completes the proof of the theorem because then $P_{p}(c_i^{p}(t^{(i)}) = Y) \leq P_{p}(c_i^{1/2}(t^{(i)}) = Y) \leq   p  $.

           \begin{lemma}~\label{lemma:a200}  $c_i^p$ and  $\mathring{c}_i^p$ are also monotone in their inputs and in $p$. \end{lemma}
           This completes the proof of Theorem~\ref{thm:non-adpative-p}. \qed
       \end{proof}

Theorems~\ref{thm:non-adpative-p} and \ref{thm:cloud-100} imply that adaptive schedules can be arbitrarily more powerful than nonadaptive ones in the strategic setting.
\begin{corollary}
For any $\epsilon > 0$, there exists a game $Q$ with strategic agents for which an adaptive scheduler can achieve $100\%$ adoption and a nonadaptive scheduler achieves $\leq \epsilon \%$.
\label{cor:adaptive-greaterthan-nonadaptive}
\end{corollary}

Similarly we see that adaptive schedules can be arbitrarily more powerful than nonadaptive ones in the myopic setting by combining Theorem~\ref{thm:non-adpative-p} and a lemma from \citet[Lemma~3.1]{Chierichetti12}, which gives a graph with adaptive performance of $(100-O(\frac{1}{pn}))\%$.

\begin{corollary}
For any $\epsilon > 0$, there exists a game $Q$ with myopic agents for which an adaptive scheduler can achieve $\geq (100-\epsilon)\%$ adoption and a nonadaptive scheduler achieves $\leq \epsilon \%$.
\label{cor:adaptive-greaterthan-nonadaptive myopic}
\end{corollary}



\section{Scheduler Commitment Power}
       \label{sec:stack}
       Our model dictates that the scheduler chooses and publishes its (possibly adaptive) schedule in advance.
       This publication is a commitment to follow the schedule even in situations where, once reached, it is suboptimal.
       We refer to a scheduler who can commit in advance as \emph{Stackelberg}, after the classic economic model of imperfect competition in which a first-moving player is notably advantaged by an ability to make \emph{non-credible} threats \cite{Stackelberg11}.
       Such ability contrasts with a scheduler who is restricted to schedules that make the performance-maximizing decision in every subgame.

Whereas the power to make non-credible threats allows players to obtain strictly greater utility in some games, there are many natural games for which this power yields no advantage.
Our question is whether in this context Stackelberg scheduling ability is strictly more powerful than the ability to only make credible threats.
A priori, it is unclear how non-credible threats could aid the scheduler.
It seems that the only way to convince a node \emph{not} to choose $N$ is to threaten to surround it with an abundance of $Y$s in the case where it does choose $N$\@.
Maximizing \( Y \)s, however, aligns with the scheduler’s goal, and can only be non-credible if somehow concentrating these $Y$s lowers overall expected performance.

       \begin{figure}
         \centering
         \includegraphics[width=0.25\columnwidth]{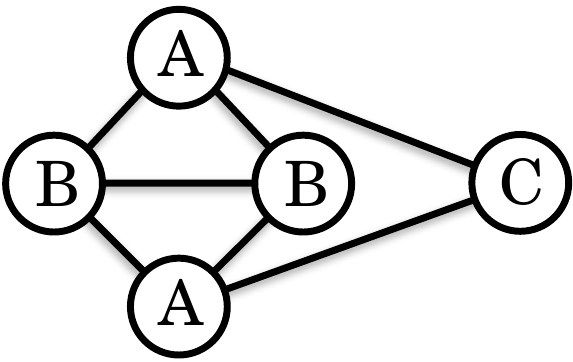}
         \caption{A graph in which a Stackelberg scheduler achieves greater performance.}
         \label{fig:3_group_graph}
       \end{figure}

       We have, however, found a game instance, illustrated in Figure~\ref{fig:3_group_graph}, for which commitment power provides an advantage.
       For this graph, with parameters $p=0.18$ and $\pi = 1.85$, a Stackelberg scheduler can achieve performance of $0.573$ whereas the best subgame-optimal schedule yields $0.371$.

       Our five node graph has three types of nodes which are in indistinguishable positions. We call the groups $A, B,$ and $C$ and don't distinguish between nodes within each group. We give the Stackelberg schedule in Figure~\ref{fig:schedule_stack} and the subgame-optimal schedule, which corresponds to a Perfect Bayesian Equilibrium (PBE), in Figure~\ref{fig:schedule_PBE}.
\begin{figure}
  \centering
  \begin{minipage}[t]{.49\textwidth}
  \subfigure[Stackelberg schedule]{
    \begin{algorithm}[H]
      \NoCaptionOfAlgo
      Schedule $A$:\\
      \uIf{$A$ chooses $Y$}
        {Schedule $A$:\\
         \uIf{$A$ chooses $Y$}
           {Schedule $B$: \\ \uIf{$B$ chooses $Y$}{Schedule $B$, then $C$}\lElse{Schedule $C$, then $B$}}
           \Else{Schedule $B$: \\ \uIf{$B$ chooses $Y$}{Schedule $C$, then $B$}\lElse{Schedule $B$, then $C$}}
        }
        \Else{Schedule $C$: \\
         \uIf{$C$ chooses $Y$}
           {Schedule $A$, then $B$, then $B$}
           \Else{Schedule $B$: \\ \uIf{$B$ chooses $Y$}{Schedule $B$, then $A$}\lElse{Schedule $A$, then $B$}}}
    \end{algorithm}
    \label{fig:schedule_stack}
  }
  \end{minipage}
  \begin{minipage}[t]{.49\textwidth}
  \subfigure[PBE schedule]{
    \begin{algorithm}[H]
    \NoCaptionOfAlgo
      Schedule $C$: \\
      \uIf{$C$ chooses $Y$}{Schedule $B$: \\
                            \uIf{$B$ chooses $Y$}
                                {Schedule $A$. All remaining nodes choose $Y$}
                            \Else{Schedule $A$, then $A$, then $B$}}
      \Else{Schedule remaining nodes in any order. All choose $N$}
    \end{algorithm}
    \label{fig:schedule_PBE}
  }
  \end{minipage}
  \caption{Schedules demonstrating the increased power of non-credible threats by the scheduler.}
  \label{fig:schedule}
\end{figure}

To see how the Stackelberg scheduler outperforms the PBE scheduler, note that the Stackelberg scheduler schedules $A$ first and it chooses its type, whereas if the PBE scheduler scheduled $A$ first it would choose $N$\@.
Both schedulers agree on what to do if $A$ chooses $N$\@.
If $A$ chooses $Y$, the Stackelberg scheduler schedules $A$ next even though scheduling $B$ next would yield higher expected performance.
The PBE scheduler must pick $B$ next in this case.
The higher performance of picking $B$ next comes at the cost of giving fewer expected $Y$-matches to $A$, and thus makes $A$, if scheduled first, less inclined to play $Y$\@.
In this instance, the result is that a first-moving $A$ would play $N$ if faced with a PBE schedule, and its type if faced with the optimal Stackelberg schedule.

Since it cannot threaten $A$, the PBE scheduler does not schedule $A$ first, and instead starts with $C$, which gives fairly similar cascade behavior but results in fewer nodes choosing $Y$, in expectation.
This completes our example of a graph with higher Stackelberg than PBE performance.

Assuming commitment power in the foregoing analysis simplifies our arguments by avoiding the necessity of verifying optimal scheduling in all subgames.
Results for myopic outperforming strategic hold \emph{a fortiori} if we relax the assumption of commitment power, as the ability to make threats is useful only for strategic agents.
Our derivation (Section~\ref{sec:cloud}) of the bound for strategic outperforming myopic exploits commitment power, however, we have verified that a more complicated demonstration can be constructed supporting the same bound under the weaker assumption of subgame-optimal PBE schedules.

\section{Computational Solutions to Strategic Cascades}
       \label{sec:block}
       It is straightforward to write a program which computes, by brute force, the optimal%
       \footnote{Our algorithm calculates the optimal schedule under the assumption that the scheduler is acting according to a PBE and \emph{cannot} make empty threats.
       The main theoretical analyses of the paper assume the scheduler \emph{can} make empty threats and is acting according to a Stackelberg equilibrium. On the star and clique these two equilibrium concepts give identical optimal schedules.
See Section~\ref{sec:stack} for in-depth discussion.}
       schedule for an arbitrary graph. Node behavior can be solved by backward induction. Logically, the exponential number of possible schedules and agent type configurations makes this approach infeasible. In this section we describe an approach to efficiently find solutions for strategic cascade problems on a subclass of highly symmetric graphs.

       Our code%
              \footnote{
              The work in this section was performed in collaboration with Erik Brinkman.
              Code can be found at \url{https://github.com/tbmbob/block-scheduling}.
              \texttt{block\_dp.py} is the file containing the solver. Code has not been prepared for public release. Please contact \url{travisbm@umich.edu} with any questions.}
       finds the optimal schedule for arbitrary \emph{blockmodel}%
              \footnote{Any graph can be expressed as a blockmodel graph with $n$ blocks. Our algorithm works efficiently only for graphs with a small number of blocks.}
       graphs with strategic or myopic agents.
       Blockmodels have been studied extensively in the past \cite{ Snijders97, Wang87} as a natural framing of networks in which nodes can be divided into classes or types with shared characteristics.
       For example, a blockmodel describing a social network at a high school could have a type for each grade. Students would be more likely to have edges to students of their same grade, and less likely to have edges to students of other grades. More abstractly, the star graph is easily described as a blockmodel in which the two classes are ``interior agent'' and ``exterior agent''.

       Our code finds optimal schedules efficiently for graphs with a small (constant) number of types. The star, clique, and cloud graphs all fit this requirement. Running time is polynomial in the number of agents and exponential in the number of blocks. The code solves for the optimal performance through a combination of dynamic programming and backward induction. It first solves all possible scenarios with one node left to choose and stores the results. Then, by using these results, the program solves optimal behavior when there are two nodes left to choose. It continues this process until it solves for the optimal behavior with all nodes left to choose. It avoids the exponential running time of a naive backward induction by treating all agents within a block the same. It is then able to consider only which \emph{block} to schedule next, not which node to schedule next. By reasoning over blocks instead of node types, the scheduler needs only to compare between $O(b)$ choices at each step, where $b$ is the number of blocks, instead of $O(n)$ choices.

       Examples of data gathered from our code can be viewed in Figures~\ref{fig:clique-simulation} and \ref{fig:star-simulation}.
       This computational method is far from a panacea. Very few real-world graphs follow strict block models, and even idealized graphs often have too many types to permit efficient simulation. However, this code has been useful in verifying results for simple star, clique, and cloud graphs and in suggesting further results. For example, simulation on the clique suggested the possibility of certain parameter spaces resulting in higher strategic performance than myopic performance.


\section{Conclusion}
       We have demonstrated that the common assumption of myopic decision making by agents participating in cascades can have significant consequences.
       For the specific model of Chierichetti et al., we find that assuming strategic instead of myopic agent decision making leads to markedly different cascade behavior.
       We show, by counterexample, that their result of linear performance for any graph does not apply when agents are strategic.
       We have identified graphs for which the performance difference between myopic and strategic agents is (asymptotically) as large as possible, in either direction.
       More broadly, we illustrate methods for reasoning about strategic cascade behavior and characterize the contrasting behavior of strategic and myopic agents in a range of qualitatively distinct settings.
       Lastly, we prove some results for strategic agents on general graphs,
and demonstrate the power of scheduler commitment.

       Modeling cascades with perfectly strategic agents is not necessarily more realistic than modeling agents with limited rationality.
       Thus, we do not argue for the strategic behavior we characterize as a definitive predictive model.
       Rather, our point is to demonstrate the potential impact of alternative assumptions about agent decision making on networks.
       It is likely that typical network decision making lies somewhere between myopic and strategic, and by characterizing the behavioral poles we hope to provide guidance for understanding the range within.
       Of course, substantial work remains to achieve a full understanding of behavior between these poles.

       We consider cascades to be representative of a broader class of scenarios involving dynamic decision on networks.
       For these too we should expect the spectrum of behaviors, myopic to strategic, to exhibit qualitative variety in generated outcomes.

\bibliographystyle{ACM-Reference-Format-Journals}
\bibliography{references}



\appendix
\setcounter{section}{0}

\section{Clique Appendix}
\label{app:clique}
\begin{proof}[of Theorem \ref{thm:clique-classes}]
              The proof of this theorem follows naturally from several lemmas which we prove in the rest of this section.  We outline the proof here.

               There are two cases of behavior to consider: $\pi < 1$ and $\pi \geq 1$.  In the first case a simple backward induction argument shows that all cliques are in $\classone$ (Lemma~\ref{lem:clique-pi-lessthan-one}).  The second case is more involved.
              We first show that if $K_n \in (\classone \cup \classtwo)$ for some $n$, $p$, and $\pi$, then all larger cliques must also be in $(\classone \cup \classtwo)$ for the same $p$ and $\pi$ (Lemmas~\ref{lem:clique-class-1} and \ref{lem:clique-class-2}).  Finally, we show that any $p$ and $\pi$ combination eventually gives a $K_n \in (\classone \cup \classtwo)$ for some large enough $n$
         (Lemma~\ref{lemma:clique-class-exists}).  The intuition for this final lemma is that, for very large cliques, the cost of ending up on the wrong side of a cascade is very large. Thus, agents always prefer to join a guaranteed cascade over choosing their type \emph{if} choosing their type has some probability of being on the wrong side of the cascade.

        Let $d = m_Y - m_N$ denote the difference between the current number of $Y$ decisions and the current number of $N$ decisions. We begin by addressing clique behavior in the first case, $\pi < 1$.
              \begin{lemma}
                  \label{lem:clique-pi-lessthan-one}
                  If $\pi < 1$ then every clique is in $\classtwo$.
              \end{lemma}
              \begin{proof}
                     Consider the behavior of the last scheduled node. It chooses its type only if $d=0$.
                     If $d \neq 0$, then it receives at least 1 more utility for choosing the majority, but only $\pi < 1$ more utility for choosing its own type.
                     Now by induction, assume that agents after time $\tau > 1$ choose their type if $d = 0$ and otherwise choose the current majority.  We must show that the node $\alpha$ at time $\tau$ will do the same.

                     If $d \geq 2$ then no matter what $\alpha$ chooses, by the inductive hypothesis, the rest of the nodes will be in a $Y$-cascade.  Thus $\alpha$ will receive at least 2  more utility for choosing the majority ($Y$), but only $\pi$ more utility for choosing its type, so it will always choose $Y$.

                     If $d = 1$, then if $\alpha$ chooses $Y$, it will cause a $Y$-cascade and receive $\tau-1$ utility for agreement with currently undecided nodes and receive $1$ more utility for its agreement with currently decided nodes.  If $\alpha$ chooses $N$, then with probability $1-p$ the next node will cause an $N$-cascade, but with probability $p$ the next node will cause a $Y$-cascade.  In the former case $\alpha$ receives $\tau -1$ utility for its agreement with currently undecided nodes.  In the latter case, $\alpha$ receives 0 utility for its agreement with currently undecided nodes.  $\alpha$'s expected payoff from agreement with currently undecided nodes for choosing $N$ is $(1 -p) (\tau-1)$.  Without considering payoff from choosing its type, $c_\alpha = Y$ yields $1 + p(\tau-1)$ more utility than $c_\alpha = N$. So, no matter the value of $t_\alpha$, $c_\alpha = Y$.

                     If $d = 0$, the inductive hypothesis gives that $c_\alpha$ starts a cascade of $\alpha$'s choice.  Thus by playing $c_\alpha = t_\alpha$, it gets $\pi$ additional utility.  The analysis for $d = -1$ and $d \leq -2$ are analogous to the cases already covered. \qed
              \end{proof}

              For the remainder of the section we assume $\pi \geq 1$. Additionally, we refer to agents on the clique according to when they are scheduled. On $K_{n}$, we call the first scheduled node $n$ and the last scheduled node 1. We begin by showing that, as $n$ increases, cliques that enter into $(\classone \cup \classtwo)$ stay that way.

              \begin{lemma}
              \label{lem:clique-class-1}
              For large enough cliques such that $(n+1)(1-p)^2 > \pi$, if $K_{n-1} \in (\classone \cup \classtwo)$ and $K_n \in \classone$ (a predetermined $N$-cascade), then $K_{n+1} \in (\classone \cup \classtwo)$.
              \end{lemma}
              \begin{proof}
              We prove this by cases, depending on how $n$ behaves if $d=1$. A diagram of behavior we know by assumption or can readily infer is shown in Figure~\ref{fig:clique-class1}.

                  \begin{figure}
                         \centering
                         \subfigure[$N$-type decision]{
                         \begin{tabular}{|c|c c c|} \hline
                           $d$: & -1 & 0 & 1 \\ \hline
                           $n$  & $N$  & $N$ & $?$ \\
                           $n+1$& $N$  & $N$ & $N$ \\ \hline
                         \end{tabular}
                         }
                         \subfigure[$Y$-type decision]{
                         \begin{tabular}{|c|c c c|} \hline
                           $d$: & -1 & 0 & 1 \\ \hline
                           $n$  & $N$  & $N$ & ? \\
                           $n+1$& $N$  & $?$ & ? \\ \hline
                         \end{tabular}
                         }
                         \caption{$\classone$ behavior}
                         \label{fig:clique-class1}
                  \end{figure}
                  We characterize $c_{n+1}$ if it is $t_{n+1}$, conditional on $c_n$ when $d=1$.

                  \emph{Case 1}: $c_n = Y$ if $d=1$ and $t_n = N$. Then it must be the case that $K_{n-1} \in \classtwo$, so this results in a $Y$-cascade. Thus $c_{n+1} = t_{n+1}$ and $K_{n+1} \in \classtwo$.

                     \emph{Case 2}:  $c_n = N$ if $d=1$ and $t_n = N$. This corresponds to $\classone$. Consider the possible payoffs. $c_{n+1} = Y$ gives a $(1-p)^2$ probability of $c_{n-1} = c_{n} = N$. This causes a $N$-cascade, by the assumption that $K_{n-1} \in (\classone \cup \classtwo)$, and results in only $\pi$ payoff:
                         \[\text{Payoff for } N = n+1\]
                         \[\text{Payoff for } Y \leq (1 - (1-p)^2) (n+1+\pi) + (1-p)^2 \pi = n + 1 + \pi - (1-p)^2 (n+1)\]
                         So in case 2, $c_{n+1} = N$, regardless of type, if $(n+1)(1-p)^2 > \pi$. Thus $K_{n+1} \in \classone$.
                  \qed
              \end{proof}

              \begin{lemma}
              \label{lem:clique-class-2}
              For large enough cliques such that $(n+1)(1-p)^2 > \pi$, if $K_{n-1} \in (\classone \cup \classtwo)$ and $K_n \in \classtwo$ (the first agent, $n$, chooses $t_n$ and starts a $t_n$-cascade), then $K_{n+1} \in (\classone \cup \classtwo)$.
              \end{lemma}
              \begin{proof}
                  We prove this by cases, conditional on $c_n$ and $d$. A diagram of behavior we know by assumption or can readily infer is shown in Figure~\ref{fig:clique-class2}.

                  \begin{figure}
                         \centering
                         \subfigure[$N$-type decision]{
                         \begin{tabular}{|c|c c c|} \hline
                           $d$: & -1 & 0 & 1 \\ \hline
                           $n-1$& $N$  & $N$ & $Y$ \\
                           $n$  & $N$  & $N$ & ? \\
                           $n+1$& $N$  & $N$ & ? \\ \hline
                         \end{tabular}
                         }
                         \subfigure[$Y$-type decision]{
                         \begin{tabular}{|c|c c c|} \hline
                           $d$: & -1 & 0 & 1 \\ \hline
                           $n-1$& $N$  & ? & $Y$ \\
                           $n$  & ?  & $Y$ & $Y$ \\
                           $n+1$& ?  & ? & ? \\ \hline
                         \end{tabular}
                         }
                         \caption{$\classtwo$ behavior}
                         \label{fig:clique-class2}
                  \end{figure}
                  We characterize $c_{n+1}$ if $t_{n+1}=Y$, conditional on $c_n$ and $d$.

                  \emph{Case 1}: $c_n = Y$ if $d=1$ and $t_n = N$. Then it must be the case that $K_{n-1} \in \classtwo$, so this results in a $Y$-cascade. Thus $c_{n+1} = t_{n+1}$ and $K_{n+1} \in \classtwo$.

                     \emph{Case 2}: $c_n = N$ if $d=1$ and $t_n = N$, and $c_n = Y$ if $d=-1$ and $t_n = Y$. Then it must be the case that $K_{n-1} \in \classtwo$. The known behavior for this case is shown in Figure~\ref{fig:clique-class2-case2}. Since $c_n = Y$ if $d=-1$ and $t_n = Y$, $c_{n+1} = Y$ when $d=0$ and $t_{n+1}=Y$. Thus $K_{n+1}\in \classtwo$

                         \begin{figure}
                                \centering
                                \subfigure[$N$-type decision]{
                                \begin{tabular}{|c|c c c|} \hline
                                  $d$: & -1 & 0 & 1 \\ \hline
                                  $n-1$& $N$  & $N$ & $Y$ \\
                                  $n$  & $N$  & $N$ & $N$ \\
                                  $n+1$& $N$  & $N$ & ? \\ \hline
                                \end{tabular}
                                }
                                \subfigure[$Y$-type decision]{
                                \begin{tabular}{|c|c c c|} \hline
                                  $d$: & -1 & 0 & 1 \\ \hline
                                  $n-1$& $N$  & $Y$ & $Y$ \\
                                  $n$  & $Y$  & $Y$ & $Y$ \\
                                  $n+1$& ?  & ? & ? \\ \hline
                                \end{tabular}
                                }
                                \caption{$\classtwo$, case 2 behavior}
                                \label{fig:clique-class2-case2}
                         \end{figure}

                     \emph{Case 3}: $c_n = N$ if $d=1$ and $t_n = N$, and $c_n = N$ if $d=-1$ and $t_n = Y$. Then it must be the case that $K_{n-1} \in \classtwo$. The known behavior for this case can be seen in Figure~\ref{fig:clique-class2-case3}. By the same math as in case 2 of the proof of Lemma~\ref{lem:clique-class-1}, we get $c_{n+1} = N$, regardless of type, if $(n+1)(1-p)^2 > \pi$. Thus $K_{n+1} \in \classone$. \qed

                  \begin{figure}
                         \centering
                         \subfigure[$N$-type decision]{
                         \begin{tabular}{|c|c c c|} \hline
                           $d$: & -1 & 0 & 1 \\ \hline
                           $n-1$& $N$  & $N$ & $Y$ \\
                           $n$  & $N$  & $N$ & $N$ \\
                           $n+1$& $N$  & $N$ & ? \\ \hline
                         \end{tabular}
                         }
                         \subfigure[$Y$-type decision]{
                         \begin{tabular}{|c|c c c|} \hline
                           $d$: & -1 & 0 & 1 \\ \hline
                           $n-1$& $N$  & $Y$ & $Y$ \\
                           $n$  & $N$  & $Y$ & $Y$ \\
                           $n+1$& ?  & ? & ? \\ \hline
                         \end{tabular}
                         }
                         \caption{$\classtwo$, case 3 behavior}
                  \label{fig:clique-class2-case3}
                  \end{figure}
              \end{proof}

              Next, we show that a clique eventually falls into $\classone$ or $\classtwo$.
              \begin{lemma}
                  \label{lemma:clique-class-exists}
                  For any $p$ and $\pi$, there exists some large enough $n$ such that $K_n, K_{n+1} \in (\classone \cup \classtwo)$.
              \end{lemma}
        \begin{proof}
                     Fix $p$ and $\pi$. By backward induction one can see that there is always an $N$-cascade when $d=\lfloor -\pi \rfloor$. Call this threshold $\underline d$. Additionally, note that an $N$-type node always chooses $N$ when $d\leq 0$. This can be seen by application of Theorem~\ref{thm:non-adpative-p} and Lemma~\ref{prop:clique-monotonic-d}.

                     So, for large enough $n$, there must be an $N$-cascade when $d = \underline d +1$. This can be seen by application of Lemma~\ref{prop:clique-certain-cascade} and the observations that $d=\underline d$ gives a certain $N$-cascade but $d = \underline d + 2$ does not give a certain $Y$-cascade.

                     This argument can be repeated to show that there must be an $N$-cascade when $d = \underline d + 2$. This reasoning can be iterated until it breaks down at $d=0$. But the node choosing at $d=0$ always has the option of a certain $N$-cascade (for large enough $n$).
                     So it will choose $Y$ only if it also is faced with the option of a certain $Y$-cascade.
                     In this case, a node chooses its type $t$ and starts a cascade of that type. By the same reasoning the next node $n+1$ also either chooses $N$ or its type and starts a cascade. \qed
              \end{proof}

        \begin{lemma}
              \label{prop:clique-monotonic-d}
              Decisions are monotonic in $d$. If $c_n = Y$ for some difference $d$ and type $t$, then $c_n = Y$ for difference $d+1$ and the same type. Similarly, if $c_n = N$ for some difference $d$ and type $t$, then $c_n = N$ for difference $d-1$ and the same type.
              \end{lemma}
              \begin{proof}
                     Assume, without loss of generality, that agent $n$ is scheduled to make a decision and $d > 0$, $t_n = N$\@. Let $q_{Y}=1$ be the probability of a $Y$-cascade if $c_n = Y$ and $q_N<1$ be the probability of an $N$-cascade if $c_n = N$\@. Also assume that, if an $N$-cascade does not occur, then a $Y$-cascade occurs and $n$ gets only some constant payoff $D$\@.
                     Then consider $n$'s payoffs:
                     \[\text{Payoff for } N = q_N n + (1-q_N)D + \pi\]
                     \[\text{Payoff for } Y = n\]
                     For large enough $n$, the payoff for $Y$ will always surpass the payoff for $N$\@. \qed
              \end{proof}

              \begin{lemma}
              \label{prop:clique-certain-cascade}
              On a large enough clique $K_n$, a node always prefers a certain (probability~$1$) cascade over an uncertain (probability less than~1) cascade, no matter $d$, $p$, or $\pi$.
              \end{lemma}
        These two lemmas, combined with Theorem~\ref{thm:non-adpative-p}, are used to show that the range of differences for which nodes consider choosing their type shrinks as $n$ grows larger. Eventually the range shrinks enough that nodes have the option of a guaranteed cascade, completing the proof of Theorem~\ref{thm:clique-classes}. \qed
       \end{proof}

       \begin{proof}[of Theorem \ref{thm:clique-strategic-cascade}]
              We consider strategic users on a clique of size $> 1$ with $1 \leq \pi < 1+p$. We show, using backward induction, that the first agent to choose always selects its type and that all following agents select the same choice. A demonstration of the backward induction can be found in Figure~\ref{fig:clique-choice-table}.

              \begin{figure}
                     \centering
                     \begin{tabular}{|c||r|c|c|c|c|c|}
                     \hline
                     $k$ & $d$: & $-2$ & $-1$ & $0$ & 1 & 2 \\ \hline \hline
                     0 & Choice: & $N$ & $S$ & $S$ & $S$ & $Y$ \\ \hline
                      & Choice: & $N$ & $N$ & $S$ & $Y$ & $Y$ \\
                     1 & E[$N$ matches]: & 3 & 2 & $1-p$ & $1-p$ & $1-p$ \\
                      & E[$Y$ matches]: & $p$ & $p$ & $p$ & $2$ & $3$ \\  \hline
                      & Choice: & $N$ & $N$ & $S$ & $Y$ & $Y$ \\
                     2 & E[$N$ matches]: & 4 & 3 & $2$ & $1-p$ & $0$ \\
                      & E[$Y$ matches]: & $0$ & $p$ & 2 & $3$ & $4$ \\
                     \hline
                     \end{tabular}
                     \caption{The choice of a node on the clique for varying $d$, $k$, and $t$}
                     \label{fig:clique-choice-table}
              \end{figure}

              Figure~\ref{fig:clique-choice-table} displays the choice a node would make if there were $k$ undecided agents and the difference between choices already made, $d = m_Y - m_N$, is given in the top row. $N$ (resp. $Y$) means that both agents choose $N$ (resp. $Y$) no matter their type. $S$ means that the choices are Split: agents choose their type. The rows below ``Choice'' display the expected number of matches from other agents if an agent chooses $N$ or $Y$\@.

              To obtain the behavior observed for $k=1$, we need $(1-p) + \pi < 2$ and $p + \pi < 2$, or, rearranging: $\pi < 1 + p$ and $\pi < 1 + (1-p)$. Since $p < (1-p)$, this is satisfied with any $\pi < 1 + p$.

              The behavior for $k=2$ follows directly from the behavior observed for $k=1$, as all inequalities are only made looser. The behavior for $k>2$ follows inductively from the behavior for $k\leq 2$. Assuming an agent with $k-1$ remaining undecided agents chooses its type only when $d=0$, an agent with $k$ remaining undecided agents will behave similarly. The inequalities for this decision are identical to the inequalities for $k=2$ with a multiplier of $k$ on each side. Our table shows that, once the balance of choices shifts in one direction, it is in a node's best interest to choose the same way. Thus the first node to choose starts a cascade of the type that it selects.  
              \qed
       \end{proof}

       Theorem~\ref{thm:clique-strategic-cascade} has the following immediate corollary:
       \begin{corollary}
              \label{cor:clique-bound-lessthan-two}
              For any clique graph with any $1 \leq \pi < 1 + p$, in the limit of $n$, the optimal strategic performance is $1 + \frac{1}{p}(1 - p)(1 - 2p)$ times the optimal myopic performance.
       \end{corollary}
       \begin{proof}
              This follows from the fact that two $Y$ decisions are required to start a myopic $Y$-cascade and only one is required to start a strategic $Y$-cascade. Thus in the limit of $n$, the probability of a $Y$-cascade is $\frac{p^2}{p^2 + (1-p)^2} \approx p^2$ for myopic agents, a well-known bound from the ``gamblers ruin'' problem, but is $p$ for strategic agents.  The result immediately follows. \qed
       \end{proof}

              \begin{table*}              
              \tbl{Utilities for node decisions on a clique with one external neighbor choosing $Y$.}{%
              \centering
              \begin{tabular}{|c||r|c|c|c|c|c|c|c|}
              \hline
                     k & $d$: & $-3$ & $-2$ & $-1$ & $0$ & 1 & 2 & 3 \\ \hline \hline
                     0 & Choice & $S$ & $S$ & $S$ & $S$ & $S$ & $Y$ & $Y$ \\ \hline

                       & Choice         & $N$ & $S$ & $S$ & $S$ & $S$ & $Y$ & $Y$ \\
                     1 & E[$N$ matches] & 4 & $3-p$ & $2-p$ & $1-p$ & $1-p$ & $1-p$ & $\dots$ \\
                       & E[$Y$ matches] & $1+p$ & $1+p$ & $1+p$ & $1+p$ & 3 & 4 & 5 \\  \hline

                       & Choice     & $\dots$ & $N$ & $\dots$ & $S$ & $\dots$ & $Y$ & $\dots$ \\
                     2 & E[$N$ matches] & $\dots$ & 4 & $\dots$ & $2-2p$ & $\dots$ & $(1-p)(2-p)$ & $\dots$ \\
                       & E[$Y$ matches] & $\dots$ & $1+2p$ & $\dots$ & $\dots$ & $\dots$ & 5 & $\dots$ \\  \hline

                       & Choice         & $\dots$ & $\dots$ & $N$ & $\dots$ & $S$ & $\dots$ & $\dots$ \\
                     3 & E[$N$ matches] & $\dots$ & $\dots$ & $4$ & $\dots$ & $3-O(p)$ & $\dots$ & $\dots$ \\
                       & E[$Y$ matches] & $\dots$ & $\dots$ & $1+O(p)$ & $\dots$ & 5 & $\dots$ & $\dots$\\  \hline

                       & Choice         & $\dots$ & $\dots$ & $\dots$ & $N$ & $\dots$ & $\dots$ & $\dots$ \\
                     4 & E[$N$ matches] & $\dots$ & $\dots$ & $\dots$ & $4$ & $\dots$ & $\dots$ & $\dots$ \\
                       & E[$Y$ matches] & $\dots$ & $\dots$ & $\dots$ & $1+O(p)$ & $\dots$ & $\dots$ & $\dots$ \\  \hline
              \end{tabular}}
              \label{tab:clique}
              \end{table*}

\section{Council and Cloud Appendix}
\label{app:council}
\label{app:cloud}
       \begin{proof}[of Lemma~\ref{lem:council-behavior}]
              We prove this lemma by directly comparing expected utilities from Table~\ref{tab:clique}. We explicitly list utilities for direct comparison, when relevant and non-obvious. All utilities assume there is an external neighbor choosing $Y$ with probability 1. Here $d = m_Y - m_N$ and $k$ is the number of currently undecided agents in the clique.

              We can see from Table~\ref{tab:clique} that, for small enough $p$, and the appropriate $2 < \pi < 3$, the first scheduled node in the clique will choose $N$.
              We can also see that the table continues for cliques larger than size 6, as payoffs are shifting in favor of $N$ as the clique grows. \qed
       \end{proof}

 \begin{proof}[of Lemma~\ref{lem:3-chooses-type}]
              If $a (1-p) + \pi < b p$, then agent~2 will always match the choice of 3. The case in which 2's payoff is greatest for not matching 3 (and thus the lemma is hardest to satisfy) is $t_2 = N$, $c_1 = N$, and $t_3 = Y$\@. In this case 2 will obtain payoff $a + \pi + b (1-p)$ for choosing $N$, or payoff $a p + b$ for choosing $Y$\@. A comparison of these payoffs shows that 2 will choose $Y$ when $a (1-p) + \pi < b p$, which is true by assumption. So 2 will choose $Y$, matching with 3. All other combinations of types result only in a higher payoff to 2 for matching with 3. \qed
 \end{proof}
  
       \begin{proof}[of Lemma~\ref{lem:2-chooses-1}]
              Note that 3 will always match 2 because $b > b(1-p) + \pi > bp + \pi$, by assumption (these inequalities simplify to $bp > \pi$). In the worst case, 1 has chosen $Y$ and $t_2 = N$\@. In this case, 2 will obtain utility of $a + b$ for choosing $Y$ and utility of $a (1-p) + \pi + b$ for choosing $N$\@. Thus 2 will choose $Y$ if $ap > \pi$, which we assume to be true. \qed
       \end{proof}

\section{Omitted Miscellaneous Results}
  \label{app:misc}
       \begin{proof}[of Theorem~\ref{prop:unique-spe}]
              This theorem follows from the fact that the cascade scheduling problem can be expressed as a finite extensive form game with perfect information.
              When nodes are never indifferent between choices, the unique PBE can be constructed using backward induction, following \citet[Prop.~9.B.2]{Mas95}.
              When nodes resolve their indifference in a consistent way, such as choosing their type, the same backward induction still selects a unique PBE\@.
              \qed
       \end{proof}

       \begin{proof}[of Theorem~\ref{prop:remove-indifference}]
              Let $I$ be the set containing all agent-situation pairs $(i, R)$ where an agent is indifferent between its two choices. $(i, R)$ means that the situation is $R$ and the next agent to choose is $i$.
              Let $\overline I$ be the set of agent-situation pairs where agents are not indifferent.
              $\overline I$ is finite, so there must be some pair $(i^*, R^*) \in \overline I$ where $i^*$ has minimal difference between utility for $Y$ and $N$\@.
              Denote this difference $\delta$, and let $\epsilon = \abs{\frac{\delta}{2}}$.
              Then no other agent-situation pair in $\overline I$ results in a different decision in $Q'$, assuming behavior of $i$ for all $(i, R) \in I$ remains the same.
              Under $Q$, all agents in all situations in $I$ have equal utility for both choices but choose their type, by assumption.
              Under $Q'$, all agents in all situations in $I$ have $\epsilon$ greater utility for their type, thus make the same decision as in $Q$.
              \qed
       \end{proof}

       The following Lemma shows that a $Y$-type node will always choose $Y$ if an $N$-type node would have chosen $Y$ in a similar situation and will be used in the proof of Theorem~\ref{prop:monotonic-p}
       \begin{lemma}
              \label{lem:choose-own-type}
              Let $i_Y$ be an agent in situation $R_Y$ and $i_N$ be an agent in situation $R_N$, with $t_{i_Y} = Y$ and $t_{i_N} = N$. If $R_Y$ and $R_N$ are identical except for some $N$ decisions in $R_N$ may be $Y$ decisions in $R_Y$, and the nonadaptive schedule $S$ is the same for both nodes, then it is never the case that $c_{i_N} = Y$ but $c_{i_Y} = N$.
       \end{lemma}
       \begin{proof}[of Lemma \ref{lem:choose-own-type}]
              Let us compare the utilities from choosing $Y$ for $i_Y$ and $i_N$, and assume for the sake of contradiction that $i_N$ prefers $Y$ but $i_Y$ does not.
              Then $i_Y$'s utility must be greater for choosing $N$, and $i_N$'s utility must be greater for choosing $Y$\@.
              Below we let $\#Y(U)$, $\#N(U)$ be the expected number of $Y$, $N$ decisions, respectively, in the set of agents $U \subseteq V$ at the end of the game. The expected value $E$ is over the randomness of agent types.

              Utility comparison for $i_Y$:
              \[\pi + E(\#Y(\nb(i_Y)) \mid c_{i_Y} = Y) <  E(\#N(\nb(i_Y)) \mid c_{i_Y} = N).\]

              Utility comparison for $i_N$:
              \[E(\#Y(\nb(i_N)) \mid c_{i_N} = Y) > \pi + E(\#N(\nb(i_N)) \mid c_{i_N} = N).\]

              Between these two inequalities only $\pi$ has moved, and the conditions of $R_Y$ favor the left side of the first inequality.
              Facing the same schedule, both inequalities cannot be true, thus we have reached a contradiction.
              \qed
       \end{proof}

       \begin{proof}[of Theorem \ref{prop:monotonic-p}]
              Generate the type distribution for $Q'$ in the following way.
              Independently draw $n = \abs{V}$ types from $\{Y, N\}$, selecting $Y$ with probability $p$ and $N$ with probability $1-p$.
              This gives us a vector of $Q'$s \emph{base types}, $\mathbf{t} = (t_1,\dotsc, t_n)$.
              Next generate a vector of $Q'$s \emph{true types}, $\mathbf{t'} = (t_1',\dotsc, t_n')$, by switching $N$-types to $Y$-types with probability $(p'-p)/(1-p)$.
              \begin{displaymath}
                     t_i' = \left\{
                            \begin{array}{ll}
                                   Y & : t_{i} = Y \\
                                   Y \text{, with probability }\frac{p'-p}{1-p}& : t_i = N \\
                                   N & : \text{otherwise}
                            \end{array}
                            \right.
              \end{displaymath}
              This results in each agent in $Q'$ being $Y$-type with independent probability $p'$, as desired.
              However, now each random draw of base types for $Q'$, $\mathbf{t}$, can be coupled with a corresponding draw of types for $Q$, $\mathbf{s}$. The true types of $Q'$, $\mathbf{t'}$, have $Y$s in the same places as $\mathbf{s}$ but with some additional $N$s turned to $Y$s. By Lemma~\ref{lem:choose-own-type} one can see that $\mathbf{t'}$ results in at least as many $Y$ choices. \qed
       \end{proof}

           \begin{proof}[of Lemma~\ref{lemma:a200}]
                  First, we note that $\hat{c}_i^p$ is monotone in its inputs and in $p$ by Lemma~\ref{prop:monotonic-p-situation}.

                  In the base case, we have that $c_1^p(t_1) =\mathring{c}_1^p(t_1) = \hat{c}_1^p(t_1)$ which is monotone in the inputs and in $p$ because $\hat{c}_1^p$ is.

                  Assume that the statement is true for all $j < i$.
                  Note that $c_i^p(t^{(i)}) = \hat{c}_i^p(\mathring{c}_{i-1}^p(t^{(i-1)}), t_i)$.  Thus $\hat{c}_i^p$ is monotone in inputs and $p$  because we know $\hat{c}_i^p$ is and $\mathring{c}_{i-1}^p$ is by induction.  Also $\mathring{c}_i^p(t^{(i)}) = \mathring{c}_{i-1}^p(t^{(i-1)})\circ \hat{c}_i^p(\mathring{c}_{i-1}^p(t^{(i-1)}), t_i)$.  So $\mathring{c}_i^p$ is monotone in inputs and $p$ because we know $\hat{c}_i^p$ is and $\mathring{c}_{i-1}^p$ is by induction.
           \end{proof}

       \begin{lemma}
       \label{prop:monotonic-p-situation}
              Let $R_Q$ be a situation in game $Q = (G, p, \pi)$ and $R_{Q'}$ be a situation in game $Q' = (G, p', \pi)$ with $0\nobreak <\nobreak p\allowbreak<\nobreak p'\nobreak <\nobreak .5$. Let $R_Q$ and $R_{Q'}$ have the same nonadaptive schedule and let $R_Q$ and $R_{Q'}$ be identical except for some $N$ decisions in $R_Q$ may be $Y$ decisions in $R_{Q'}$. If the next scheduled agents in both games are the same type, and if the agent chooses $Y$ in game $Q$, then the agent in game $Q'$ also chooses $Y$.
       \end{lemma}

       \begin{proof}[of Lemma~\ref{prop:monotonic-p-situation}]
    The proof of this lemma proceeds by a coupling of the unscheduled agents of $Q$ and $Q'$. First, each individual agent in $Q'$ is at least as likely to be $Y$-type as a corresponding agent in $Q$.  By Lemma~\ref{lem:choose-own-type}, each individual agent is more likely to choose $Y$ in $Q'$ as in $Q$, and thus the current agent will only choose $Y$ in $Q'$ if it would in $Q$.
    \end{proof}

\section{Omitted star results}
  \label{app:star}
  \label{sec:star-opt-analysis}
  Section~\ref{sec:star-opt-analysis} is devoted to proving Theorem~\ref{thm:star-spe}, which states that performance is always greater with myopic agents than strategic agents on the star, no matter the situation. A sketch is as follows.

  We break the proof into two cases depending on the value of $\pi$. We first handle the case where $\pi \geq 1$ in Theorem~\ref{theorem star pi big}. We next consider the case where $\pi < 1$ and the scheduler is limited to nonadaptive schedules in Theorem~\ref{prop:star-nonadaptive-little-pi}.
  The final case, where $\pi <1$ with an adaptive schedule, is more involved.
  We first give a schedule, $S_{opt}$, and show that it is weakly optimal in both the strategic and myopic settings in Lemma~\ref{lemma:star schedule little pi}.
  We prove optimality by showing that scheduling the interior node before a $Y$ majority has been reached is never a better option.
  Lastly, we show higher myopic performance under this optimal adaptive schedule by detailing the behavior of agents under this schedule.

  We begin by analyzing the case $\pi \geq 1$.
  \begin{theorem} \label{theorem star pi big}
  For any star graph with $\pi \geq 1$, for \emph{any} schedule, the myopic performance is greater than the strategic performance.
  \end{theorem}
  \begin{proof}
    Game behavior is simple when $\pi \geq 1$: every exterior agent, strategic or myopic, chooses its type.
    It is only the interior agent's behavior that \emph{might} differ.

    For any fixed draw of agent types $\mathbf{t} = (t_1,\dotsc, t_n)$ with schedule $S$, myopic and strategic agents behave identically except for the interior agent~$i$.
    Thus, for any fixed $\mathbf{t}$, the situations in which $S$ selects $i$ to decide next for myopic agents are the same as for strategic agents.
    A strategic $i$ reasons that each of its undecided exterior neighbors will choose $Y$ with probability $p < .5$, and therefore expects more of its undecided neighbors to choose $N$ than $Y$\@.
    A myopic interior agent ignores this fact and chooses as if equal numbers of undecided neighbors will choose $N$ and $Y$\@.
    Thus a strategic interior agent is less likely to choose $Y$ because it shifts its expected payoff in favor of $N$\@. \qed
  \end{proof}

  Next, we consider the case $\pi < 1$.
  We first give a cursory examination of agent behavior.
  A myopic exterior agent chooses its type if the interior agent has not yet chosen.
  If the interior agent has chosen $c$ and an exterior agent (strategic or myopic) is scheduled to decide, it will also choose $c$ because $\pi < 1$ guarantees that an agent prefers a matching choice over choosing its type.
  Knowing this, a strategic interior agent always chooses the majority choice of the already decided agents and breaks ties with its type.
  A myopic interior agent behaves the same way.

  We summarize the behavior of myopic exterior agents in a formal theorem, for reference below.
  \begin{theorem}\label{theorem:myopic little pi strategies}              For any star graph with $\pi < 1$, myopic exterior agents choose their type if scheduled before the interior node and match the choice of the interior node if scheduled after it.  \end{theorem}

  Strategic and myopic agents differ only in the behavior of exterior agents scheduled before the interior agent has decided.
  Myopic agents choose their type, as noted above, but strategic agents choose based on what they expect the interior agent to choose.
  A strategic $Y$-type exterior agent might choose $N$ if it sees that many exterior agents have already chosen $N$, and thus that the interior agent is likely to choose $N$\@.
  This assessment, however, depends on the schedule.

  We first address the case where the scheduler is limited to nonadaptive schedules.
  \begin{theorem}\label{prop:star-nonadaptive-little-pi}
  For any star graph with $\pi < 1$, an optimal nonadaptive schedule selects the interior agent first for both strategic and myopic agents.
  \end{theorem}
  \begin{proof}
  If the interior agent is scheduled first, it will choose its type and the exterior agents will follow.  This happens for both myopic and strategic agents, and guarantees a performance of $p$.  By Theorem~\ref{thm:non-adpative-p} this is the best possible performance for a nonadaptive schedule.  Thus choosing the interior agent first is an optimal nonadaptive schedule and, in this case, the strategic and myopic performance is identical. \qed
  \end{proof}

  We next define the adaptive schedule $S_{opt}$ and prove its optimality for both strategic and myopic agents.
  \begin{definition}
  Schedule $S_{opt}$ is the following:
  \begin{enumerate}
  \item Schedule exterior agents until a majority decide $Y$ or all have decided, whichever comes first.
  \item Schedule the interior agent.
  \item Schedule all remaining exterior agents.
  \end{enumerate}
  \end{definition}

  The optimality of $S_{opt}$ is summarized in the lemma below.
  \begin{lemma} \label{lemma:star schedule little pi} For any star graph with $\pi < 1$,  $S_{opt}$ is a weakly optimal adaptive schedule for both strategic and myopic agents.\end{lemma}
  \begin{proof}
    The combination of Theorems~\ref{prop:choose-ext-first} and \ref{prop:star-always-schedule-exterior} shows that $S_{opt}$ is weakly optimal for myopic and strategic agents (establishing Lemma~\ref{lemma:star schedule little pi}).
    It is only weakly optimal: for some parameter settings other schedules give equal performance.
    For example, when $\pi$ is sufficiently small, a strategic exterior $Y$-type node will choose $N$ when $d=-1$, and thus scheduling the interior node first yields equal performance.
  \end{proof}

  \begin{theorem}
  \label{prop:choose-ext-first}
  For both strategic and myopic agents on the star graph, $S_{opt}$ gives no worse performance than scheduling the interior agent first.
  \end{theorem}

  \begin{proof}
    For both strategic and myopic agents, scheduling the interior agent first guarantees all exterior agents will match its decision.
    So all agents choose $Y$ with probability $p$ and $N$ with probability $1-p$.
    This gives $pn$ performance.

    $S_{opt}$ calls for scheduling an exterior agent, $e$, first.
    With probability $p$, $t_e = Y$.
    If $t_e = Y$ and $c_e = Y$, then $S_{opt}$ schedules the interior agent next and all agents choose $Y$\@.
    This gives $e$ utility $\pi + 1$, its maximum possible utility.
    So strategic $e$ always chooses $Y$ if $t_e = Y$.
    And myopic $e$ will choose its type, so will also choose $Y$ if $t_e = Y$. This alone gives $pn$ performance, without considering outcomes for $t_e = N$. \qed
  \end{proof}

  \begin{theorem}
  \label{prop:star-always-schedule-exterior}
  For any star graph, while exterior $Y$-decisions are \emph{not in the majority}, scheduling an exterior agent results in performance at least as high as scheduling the interior agent.
  \end{theorem}

  \begin{proof}
    Let $d = m_Y - m_N$ denote the difference between the current number of $Y$ decisions and the current number of $N$ decisions.
    First consider the case $d < 0$.
    If scheduled, the interior node will choose $N$ and all remaining exterior nodes will choose $N$\@.
    This results in zero additional $Y$-adoptions, so scheduling an exterior node instead, as $S_{opt}$ does, must be at least as good.

    The other possibility is $d=0$, which reduces to the situation covered in the proof of Theorem~\ref{prop:choose-ext-first}. \qed
  \end{proof}

  Now that we've shown $S_{opt}$ is an optimal adaptive schedule for strategic agents, we give a detailed characterization of agent behavior under $S_{opt}$ and show that this behavior cannot lead to higher strategic performance than myopic performance.

  The probability of ever getting a majority of exterior $Y$-adoptions, in the limit of large $n$, is $\frac{p}{1-p}$, a well known result in the mathematics of biased random walks.

  We seek to show that strategic agents choose $N$ in any situation where myopic agents would choose $N$\@, and thus that strategic performance is lower than myopic performance.
  By Theorem~\ref{theorem:myopic little pi strategies} and Lemma~\ref{lemma:star schedule little pi}, it is sufficient to show that, under $S_{opt}$, $N$-type exterior nodes that are scheduled before the interior node always choose $N$\@.
  To prove this we start by characterizing these agents' strategies.

  For the remainder of the section, we omit the word ``exterior'' when it is clear from context and refer to the interior node as~$i$.
  For a given situation, let $d$ be the difference between the number of exterior nodes who have chosen $Y$ and the number who have chosen $N$: $d = m_Y - m_N$\@.
  A low $d$ value indicates that more nodes have chosen $N$ and indicates a higher likelihood of $c_i = N$\@.
  Thus nodes are more inclined to choose $N$ for low values of $d$ and more inclined to choose $Y$ for high values of $d$.

  Let $i(d, k)$ denote the probability that node~\( i \) will choose $Y$, as assessed from the perspective of an exterior node in a situation with difference $d$ and $k$ unscheduled exterior nodes.
          Let $u(t, c, d, k)$ denote the expected utility of a type $t$ node, making choice $c$, with a difference of $d$ when there are $k$ unscheduled exterior nodes.
  We begin counting at $1$, so $k = 1$ refers to the choice of the final unscheduled exterior node.

  We say that exterior nodes execute a \emph{threshold strategy} if there exists  thresholds $Y^*(k)$ and $N^*(k)$ such that:
  \begin{itemize}
  \item $d < Y^*(k)$: all agents choose $N$ when $k$ nodes remain to choose.
  \item $Y^*(k) \leq  d \leq N^*(k)$: agents choose their type  when $k$ nodes remain to choose.
  \item $ N^*(k) < d$: all agents choose $Y$ when $k$ nodes remain to choose.
  \end{itemize}

  \begin{figure}
  \centering
  \begin{tabular}{|c|c c c c c c c|} \hline
  $d$: & $-\infty$ & $\dotsc$ & $Y^*(k)$ & $\dotsc$ & $N^*(k)$ & $\dotsc$ & $1$\\
          \hline \hline
          $t_e = N$ & $N$ & $N$ & $N$ & $N$ & $N$ & $Y$ & $Y$\\ \hline
          $t_e = Y$ & $N$ & $N$ & $Y$ & $Y$ & $Y$ & $Y$ & $Y$\\ \hline
          \end{tabular}
          \caption{Threshold behavior for exterior nodes, by type.}
          \label{fig:star-thresholds}
          \end{figure}

          An illustration of these thresholds can be seen in Figure~\ref{fig:star-thresholds}.

  We start out with the following theorem, which shows that the interior agent is always more likely to choose $Y$ if scheduled immediately instead of after one additional exterior node. This theorem serves as the base case for inductive arguments in Theorem~\ref{prop:starthreshold} and Lemma~\ref{lemma:starlittlepimain}, defined below.
  \begin{theorem}  \label{prop:basecasestar} 
    $i(d, 1) \geq i(d, 0)$ for all $d$, and both are monotone increasing in $d$.  Moreover, the final exterior agent and the interior agent (if scheduled last) play threshold strategies where $N^*(1) = N^*(0) = 1$.
  \end{theorem}

  \begin{proof}
      Behavior of exterior agents, and thus the value of $i(d, k)$, falls into two classes: $\pi + p < 1$ and $\pi + p \geq 1$. Our proof proceeds by finding $i(d, 0)$ (the probability that the interior node chooses $Y$ when scheduled immediately) in both cases, and then examining $i(d, 1)$ (the probability that the interior node chooses $Y$ when scheduled after one additional exterior node) for each case individually.

      In either case, the interior node chooses its type when it is scheduled only if $d = 0$.
      It will choose $Y$ if $d > 0$ and $N$ if $d < 0$.  Thus: $$i(d, 0) = \left\{\begin{array}{cc} 0 & d \leq -1 \\ p & d = 0 \\ 1 & d \geq 1  \end{array} \right.$$

      In the case that $\pi + p < 1$, $i(d, 1) = i(d,0)$.  We analyze $i(d,1)$ by considering the behavior of the exterior node scheduled immediately before the interior node. We call this final exterior node $e$. If $d = 0$ and $k=1$, the interior node will match $c_e$.
    Knowing this, $e$ should choose its type if $d=0$ and $k=1$.

    If $d = -1$ then $e$ should always choose $N$\@.
    If $t_e = N$ this will guarantee a payoff of $1 + \pi$, the maximum possible.  If $t_e = Y$, then selecting $N$ yields payoff 1.  Selecting $Y$ yields $1 + \pi$ if the interior node is $Y$-type, but only $\pi$ if the interior node is $N$-type.  This gives expected payoff of $p + \pi < 1$.

      In the case that $\pi + p \geq 1$, similar analysis shows that:  $$i(d, 1) = \left\{\begin{array}{cc} 0 & d \leq -2 \\ p^2 & d = -1 \\ p & d = 0 \\ 1 & d \geq 1 \end{array} \right.$$  The theorem follows.
  \end{proof}

  We use the below theorem to show that nodes always behave according to a threshold strategy.
  \begin{theorem} \label{prop:starthreshold} For any star graph with $\pi < 1$, under schedule $S_{opt}$, exterior nodes scheduled before the interior agent execute a threshold strategy.  \end{theorem}
  \begin{proof}
      The proof proceeds by induction.  At $k= 0$, by Theorem~\ref{prop:basecasestar}, the interior agent plays a threshold strategy.  Assuming that after some node $e$ all agents play a threshold strategy, we show that $e$ does also.

      Assume without loss of generality that $t_e = Y$. Assume, for the sake of contradiction, that $e$ plays a non-threshold strategy. This necessitates $c_e = Y$ for some $d$ but $c_e = N$ for some $d+1$.
      In this case, $e$ changes its choice to $N$ if an additional node has chosen $Y$\@.
      However, because all later scheduled nodes play according to threshold strategies, an additional node choosing $Y$ only increases the chance that the interior node chooses $Y$, so $e$'s strategy is not rational. This gives a contradiction, as desired. \qed
  \end{proof}

  Knowing that nodes behave according to a threshold strategy, we can prove the following lemma, which is enough to complete the proof of Theorem~\ref{theorem star}, as discussed above.
          \begin{lemma}  \label{lemma:starlittlepimain} For any star graph with $\pi < 1$, under schedule $S_{opt}$, exterior nodes scheduled before the interior always choose $N$ if they are $N$-type. \end{lemma}
  \begin{proof}{Lemma \ref{lemma:starlittlepimain}}
      The proof of this lemma proceeds by induction on a slightly stronger statement.  By induction on $k$ we show that $N^*(k) = 0$ (which implies the lemma), and that $Y^*(k)$ is monotone decreasing in $k$.  The base case is covered by Theorem~\ref{prop:basecasestar}. 

      We now prove the inductive step.  Let $k$ denote the the number of unscheduled exterior nodes, and assume that for all $k < \ell$, $N^*(k) = 0$ and $Y^*(k)$ is monotone decreasing in $k$.  We now prove the statement for $k = \ell$.

      Note that $u(Y,N,Y^*(\ell-1),\ell) = 1 = u(Y,N,Y^*(\ell-1),\ell - 1)$.  By induction, we know that $Y^*(k)$ is monotone decreasing for $k < \ell$, so for any $d < Y^*(\ell - 1)$ we are in an $N$-cascade situation.
    In this case agent~$\ell$ receives a payoff of 1 for choosing $N$ because it will guarantee a match with the interior agent, but does not pick its own type.

      Also note that $u(Y,Y,Y^*(\ell-1),\ell) \geq  u(Y,Y,Y^*(\ell-1),\ell - 1)$.  This follows from a coupling argument between two situations which we define. In situation~1,  $d = Y^*(\ell-1) + 1$ and there are $\ell-1$ unscheduled exterior nodes.  In situation~2,  we have that $d = Y^*(\ell-1) + 1$  and there are $\ell-2$ unscheduled exterior nodes.

      Couple the randomness so that the type of the agent scheduled with $k$ unscheduled exterior agents in situation~1 is the same as the type of the agent scheduled with $k - 1$  unscheduled exterior agents in situation~2.  Note that the randomness of the final unscheduled exterior agent has not yet been fixed in schedule 1.  Then by monotonicity of $Y^*(k)$ established by induction,
      whenever $d$ decreases in situation~1, it also decreases in situation~2 (though the converse is not necessarily true).   The $Y$ surplus $d$ of situation~1 with 1 unscheduled exterior node (\( k = 1 \)), is at least the $Y$ surplus $d$ of situation~1 with 0 unscheduled exterior nodes (\( k = 0 \)).
      By Theorem~\ref{prop:basecasestar} we have that the probability that the interior agent chooses $Y$ in situation~1 is at least the probability that the interior agent chooses $Y$ in situation~$2$.

      Finally, note that  $u(Y,Y,Y^*(\ell-1),\ell-1) \geq u(Y,N,Y^*(\ell-1),\ell-1)$ because, by the definition of $Y^*(\ell - 1)$  in this situation, $Y$-type agents choose $Y$.

      Putting this together we get that $u(Y,Y,Y^*(\ell-1),\ell) \geq u(Y,Y,Y^*(\ell-1),\ell-1) \geq u(Y,N,Y^*(\ell-1),\ell-1) = 1 = u(Y,N,Y^*(\ell-1),\ell - 1)$.  This shows that at $Y^*(\ell-1)$ with $k$ nodes left unscheduled, a $Y$-type node chooses $Y$\@.
    Thus we have that $Y^*(k)$ is monotonically decreasing.

      Lastly, we show that $N^*(\ell) = 0$.  This is equivalent to proving that a rational $N$-type agent would choose $N$ over $Y$, or that $u(N, N, 0, \ell) \geq u(N, Y, 0, \ell)$, which can also be written $(1 - i(-1, \ell - 1)) + \pi \geq 1$.
    Note that by the definition of $Y^*(\ell)$ we have that
      $i(Y^*(\ell)+1, \ell - 1) + \pi \geq 1$.  But by the Lemma~\ref{similardist} we see that $i(Y^*(\ell)+1, \ell - 1) \leq (1 - i(-1, \ell - 1))$, and so $(1 - i(-1, \ell - 1)) + \pi \geq 1$, as desired. \qed
  \end{proof}

  This lemma relies on the following theorem, which allows us to compare the probabilities of two random walks trying to reach opposite endpoints, or thresholds, by traveling similar distances. Thus, whatever the thresholds $Y^*(k)$ and $N^*(k)$ end up being, the $N$-type node at $N^*(k)$ has a higher chance of ending up in an $N$-cascade than the $Y$-type node at $Y^*(k)$. By applying the following theorem, Lemma~\ref{lemma:starlittlepimain} shows that $N^*(k) = 0$ is always a valid threshold.

          \begin{theorem}
          \label{similardist}
          $i(Y^*(\ell) + 1, \ell -1) \leq 1 - i(-1, \ell - 1).$\end{theorem}

\begin{proof}
    The proof follows from a coupling argument and some case analysis.  Consider two situations on the star with $\ell - 1$ undecided nodes: situation~1 with $d_1 = Y^*(\ell) + 1$ and, situation~2 with $d_2 = -1$.  Couple the randomness so that whenever a node scheduled in situation~1 is $Y$-type the corresponding node in situation~2 is $N$-type. 
    It follows that at any later step either a) situation~1 has reached a $Y$-cascade, or b)  $d_2 \leq Y^*(\ell) - d_1$.
    We show if a) is not satisfied then b) is. When a) is not satisfied, situation~1 is not in a $Y$-cascade and agents choose their type.  In this case, $d_1$ can increase only if the currently choosing agent in situation~1 is $Y$-type.  However, whenever this happens the agent in situation~2 is $N$-type and so $d_2$ decreases. This preserves the truth of b).

    The interior node in situation~1 chooses $Y$ only if 1) a $Y$-cascade is reached, or $d_1 = 0$ when there are no remaining exterior nodes ($k = 0$) and $t_i = Y$.
    If situation~1 enters a $Y$-cascade, then situation~2 enters an $N$-cascade because the former happens only if $d_1$ ever reaches 1, but then $d_2 \leq Y^*(\ell) + 1$ so that situation~2 is an $N$-cascade.

    If $d_1 = 0$ when there are no unscheduled exterior nodes remaining and $t_i = Y$ in situation~1, then in situation~2, we have that $d_2 \leq 0$ and, by coupling, that $t_i = N$, so that the interior node  always chooses $N$.

    We have shown than whenever the interior node chooses $Y$ in situation~1, the interior node chooses $N$ in situation~2.
\end{proof}

\paragraph*{\textbf{Computational star performance}}
              \begin{figure}
              \centering
              \subfigure[Varying \(p\), $\pi = 0.9$]{\includegraphics[width=0.493\columnwidth]{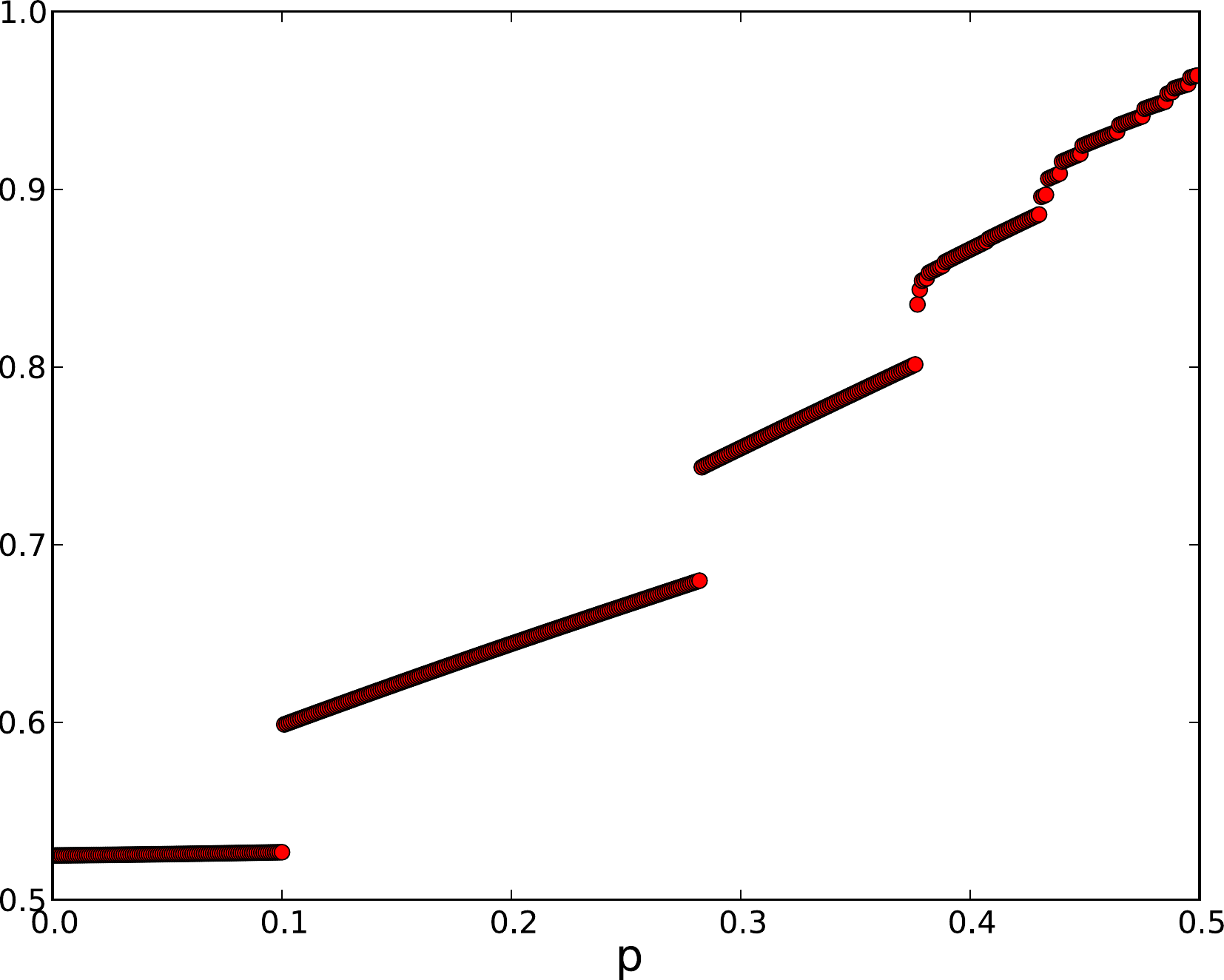}}
              \subfigure[Varying \(\pi\), $p = 0.45$]{\includegraphics[width=0.493\columnwidth]{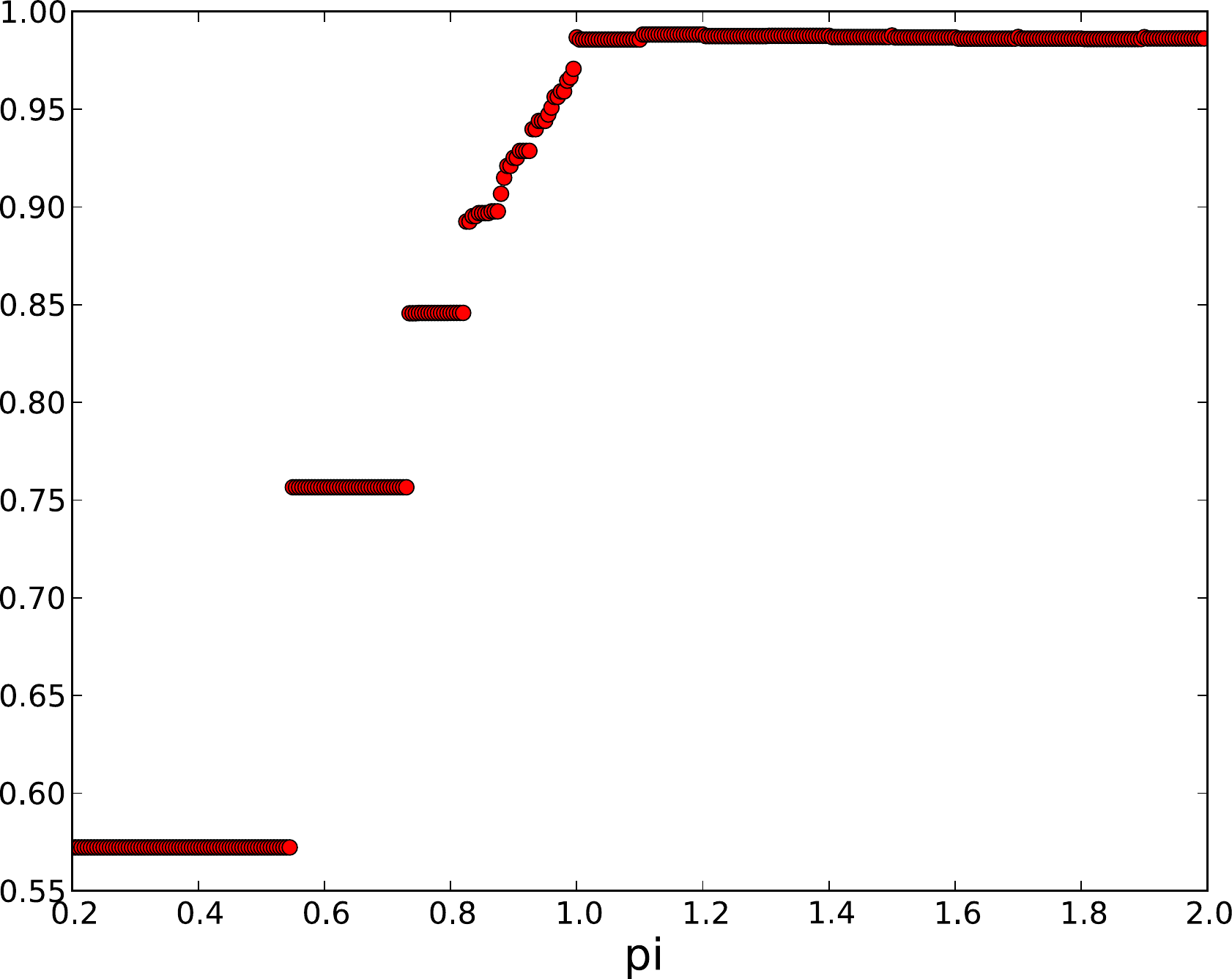}}
              \caption{Strategic-to-myopic performance ratio on star, \( n=21 \).}
              \label{fig:ratio2}
              \end{figure}
               Figure~\ref{fig:ratio2} graphs values of the strategic-to-myopic performance ratio for a star of 21 agents to show asymptotic behavior, and shows that the ratio approaches 1 but never exceeds it, as per Theorem~\ref{thm:star-spe}.

\end{document}